\documentclass[12pt, English,titlepage]{article}

\usepackage{epsfig}
\usepackage{setspace}
\usepackage{ifthen}
\usepackage{graphicx}
\usepackage{amsmath}
\usepackage{amsfonts}
\usepackage{amssymb}
\usepackage{xypic}
\usepackage{lscape}
\usepackage{amsthm}
\usepackage{xy}
\usepackage{yhmath}
\usepackage{lscape}
\usepackage{amscd}
\usepackage{paralist}
\usepackage{graphics} 
\usepackage{epsfig} 

\def\br#1{\left(#1\right)}
\def\Br#1{\left[#1\right]}
\def\BR#1{\left\{#1\right\}}

\def\angles#1{\left<#1\right>}

\newtheorem{theorem}{Theorem}
\newtheorem{defi}[theorem]{Definition}

\newtheorem{corollary}[theorem]{Corollary}

\newtheorem{definition}[theorem]{Definition}
\newtheorem{example}[theorem]{Example}

\newtheorem{lemma}[theorem]{Lemma}
\newtheorem{notation}[theorem]{Notation}

\newtheorem{proposition}[theorem]{Proposition}
\newtheorem{remark}[theorem]{Remark}

\def\br#1{\left(#1\right)}
\def\Br#1{\left[#1\right]}
\def\BR#1{\left\{#1\right\}}

\def\angles#1{\left<#1\right>}

\DeclareMathOperator{\diag}{diag} \DeclareMathOperator{\dive}{div}
\DeclareMathOperator{\Hol}{Hol}
\DeclareMathOperator{\LAhol}{\mathfrak{hol}}
\DeclareMathOperator{\spa}{Span}

\begin{document}

\title{Geometric Arbitrage Theory and Market Dynamics Reloaded}

\author{Simone Farinelli\\
        Core Dynamics GmbH\\
        Scheuchzerstrasse 43\\
        CH-8006 Zurich\\
        Email: simone@coredynamics.ch
        }

%

\bigskip
\bigskip
\bigskip

\maketitle \thispagestyle{empty} \nopagebreak

\begin{abstract}
This paper is essentially a new version of \cite{Fa15}, where some flaws have been amended.
We have embedded the classical theory of stochastic finance into a
differential geometric framework called \textit{Geometric Arbitrage
Theory} and show that it is possible to:

\begin{itemize}
\item Write arbitrage as curvature of a principal fibre bundle.
\item Parameterize arbitrage strategies by its holonomy.
\item Give the Fundamental Theorem of Asset Pricing a
differential homotopic characterization.
\item Characterize Geometric Arbitrage Theory by five principles and
show they are consistent with the classical theory of
stochastic finance.
\item Derive for a closed market the equilibrium solution for market portfolio and
dynamics in the cases where:
 \begin{itemize}
  \item Arbitrage is allowed but minimized.
  \item Arbitrage is not allowed.
 \end{itemize}
\item Prove that the no-free-lunch-with-vanishing-risk condition
implies the zero curvature condition. The converse is in general
not true and additionally requires the Novikov condition for the
instantaneous Sharpe Ratio to be satisfied.
\end{itemize}

\end{abstract}

\section{Introduction}
This paper develops a conceptual structure - called
\textit{Geometric Arbitrage Theory or GAT} - embedding the
classical stochastic finance into a stochastic differential
geometric framework. The main contribution of this approach
consists of modelling markets made of basic financial instruments
together with their term structures as principal fibre bundles.
Financial features of this market - like no arbitrage and
equilibrium - are then characterized in terms of standard
differential geometric constructions - like curvature - associated
to a natural connection in this fibre bundle or to a stochastic
Lagrangian structure that can be associated to it.\\Several
research areas can benefit from the GAT approach:
\begin{itemize}
\item Risk management, with the development of a consistent scenario
generators reducing the complexity of the market, while
maintaining the fundamental connections between financial
instruments and allowing for a reconciliation of econometric
forecasting with SDEs techniques. See Smith and Speed
(\cite{SmSp98}).
\item Pricing, hedging and statistical arbitrage, with the
development of generalized Black-Scholes equations accounting for
arbitrage and the computation of positive arbitrage strategies in
intraday markets. See Farinelli and Vazquez (\cite{FaVa12}) for a
practical application leading to an almost one probability growth
portfolios with real assets.
\end{itemize}

Principal fibre bundle theory has been heavily exploited in
theoretical physics as the language in which laws of nature can be
best formulated by providing an invariant framework to describe
physical systems and their dynamics. These ideas can be carried
over to mathematical finance and economics. A market is a
financial-economic system that can be described by an appropriate
principle fibre bundle. A principle like the invariance of market
laws under change of num\'{e}raire can be seen then as gauge
invariance. The fact that gauge theories are the natural language
to describe economics was first proposed by Malaney and Weinstein
in the context of the economic index problem (\cite{Ma96},
\cite{We06}). Ilinski (see \cite{Il00} and \cite{Il01}) and Young
(\cite{Yo99}) proposed to view arbitrage as the curvature of a
gauge connection, in analogy to some physical theories.
Independently, Cliff and Speed (\cite{SmSp98}) further developed
Flesaker and Hughston seminal work (\cite{FlHu96}) and utilized
techniques from differential geometry (indirectly mentioned by
allusive wording) to reduce the complexity of asset models before
stochastic modelling. Perhaps due to its borderline nature lying
at the intersection between stochastic finance and differential
geometry, there was almost no further mathematical research, and
the subject, unfairly considered as an exotic topic, remained
confined to econophysics, (see \cite{FeJi07}, \cite{Mo09} and
\cite{DuFiMu00}). We would like to demonstrate that Geometric
Arbitrage Theory can be given a rigorous mathematical background
and can bring new insights to mathematical finance by looking at
the same concepts from a different perspective. That for we will
utilize the formal background of stochastic differential geometry
as in Schwartz (\cite{Schw80}), Elworthy (\cite{El82}), Em\'{e}ry
(\cite{Em89}), Hackenbroch and Thalmaier (\cite{HaTh94}), Stroock
(\cite{St00}) and Hsu (\cite{Hs02}).\par This paper is structured
as follows. In Section 2, after an introductory review of
classical stochastic finance, the primitives of Geometric
Arbitrage Theory are explained. Section 3 develops the foundations
of GAT, allowing an interpretation of arbitrage as curvature of a
principal fibre bundle representing the market and defining the
quantity of arbitrage associated to a market or to a
self-financing strategy. The no-free-lunch-with-vanishing-risk (or
NFLVR for short) condition implies the vanishing of the curvature.
The converse is in general not true and additionally requires the
instantaneous Sharpe Ratio for the asset value dynamics to satisfy
the Novikov condition. The NFLVR condition has the interpretation
of a continuity equation satisfied by value density and current
of the market, as fluid density and current in the hydrodynamics
of an incompressible flow. If all market agents follow the principle
 of expected utility maximization, then the curvature vanishes and viceversa.
Section 4 provides a guiding example for a market whose asset
prices are It\^{o} processes. In Section 5 the connections between
mathematical finance and differential topology are analyzed.
Homotopic equivalent self-financing arbitrage strategies can be
parameterized by the Lie algebra of the holonomy group of the
principal fibre bundle. The no-free-lunch-with-vanishing-risk
condition is seen to be equivalent to the triviality of the
holonomy group or to the triviality of the homotopy group. This is
a differential-homotopic formulation of the Fundamental Theorem of
Asset Pricing. In Section 6 we express the market model in terms
of a stochastic Lagrangian system, whose dynamics is given by the
stochastic Euler-Lagrange Equations. Symmetries of the Lagrange
function can be utilized to derive first integrals of the dynamics
by means of the stochastic version of N\"other's Theorem.
Equilibrium and non-equilibrium solutions are explicitly computed.
Section 7 concludes.

\section{Geometric Arbitrage Theory Fundamentals}\label{section2}

\subsection{The Classical Market Model}\label{StochasticPrelude}
In this subsection we will summarize the classical set up, which
will be rephrased in Section \ref{foundations} in differential
geometric terms. We basically follow \cite{HuKe04} and the ultimate
reference \cite{DeSc08}.\par We assume continuous time trading and
that the set of trading dates is $[0,+\infty[$. This assumption is
general enough to embed the cases of finite and infinite discrete
times as well as the one with a finite horizon in continuous time.
Note that while it is true that in the real world trading occurs at
discrete times only, these are not known a priori and can be
virtually any points in the time continuum. This motivates the
technical effort of continuous time stochastic finance.\par The
uncertainty is modelled by a filtered probability space
$(\Omega,\mathcal{A}, \mathbb{P})$, where $\mathbb{P}$ is the
statistical (physical) probability measure,
$\mathcal{A}=(\mathcal{A}_t)_{t\in[0,+\infty[}$ an increasing
family of sub-$\sigma$-algebras of $\mathcal{A}_{\infty}$ and
$(\Omega,\mathcal{A}_{\infty}, \mathbb{P})$ is a probability space.
The filtration $\mathcal{A}$ is assumed to satisfy the usual
conditions, that is
\begin{itemize}
\item right continuity: $\mathcal{A}_t=\bigcap_{s>t}\mathcal{A}_s$ for all $t\in[0,+\infty[$.
\item $\mathcal{A}_0$ contains all null sets of
$\mathcal{A}_{\infty}$.
\end{itemize}
The market consists of finitely many \textbf{assets} indexed by
$j=1,\dots,N$, whose \textbf{nominal prices} are given by the
vector valued semimartingale
$S:[0,+\infty[\times\Omega\rightarrow\mathbf{R}^N$ denoted by
$(S_t)_{t\in[0,+\infty[}$ adapted to the filtration $\mathcal{A}$.
The stochastic process $(S^ j_t)_{t\in[0,+\infty[}$ describes the
price at time $t$ of the $j$th asset in terms of  unit of cash
\textit{at time $t=0$}. More precisely, we assume the existence of a
$0$th asset, the \textbf{cash}, a strictly positive
semimartingale, which evolves according to
$S_t^0=\exp(\int_0^tdu\,r^0_u)$, where the integrable
semimartingale $(r^0_t)_{t\in[0,+\infty[}$ represents the
continuous interest rate provided by the cash account: one always
knows in advance what the interest rate on the own bank account
is, but this can change from time to time. The cash account is
therefore considered the locally risk less asset in contrast to
the other assets, the risky ones. In the following we will mainly
utilize \textbf{discounted prices}, defined as
$\hat{S}_t^j:=S_t^j/S^{0}_t$, representing the asset prices in
terms of \textit{current} unit of cash.\par
 We remark that there is no need to
assume that asset prices are positive. But, there must be at least
one strictly positive asset, in our case the cash. If we want to
renormalize the prices by choosing another asset instead of the
cash as reference, i.e. by making it to our
\textbf{num\'{e}raire}, then this asset must have a strictly
positive price process. More precisely, a generic num\'{e}raire is
an asset, whose nominal price is represented by a strictly
positive stochastic process $(B_t)_{t\in[0,+\infty[}$, and
 which is a portfolio of the original assets $j=0,1,2,\dots,N$. The discounted prices of the original
assets are  then represented in terms of the num\'{e}raire by the
semimartingales $\hat{S}_t^j:=S_t^j/B_t$.\par We assume that there
are no transaction costs and that short sales are allowed. Remark
that the absence of transaction costs can be a serious limitation
for a realistic model. The filtration $\mathcal{A}$ is not
necessarily generated by the price process
$(S_t)_{t\in[0,+\infty[}$: other sources of information than
prices are allowed. All agents have access to the same information
structure, that is to the filtration $\mathcal{A}$.\par A
 \textbf{strategy} is a predictable stochastic
process $x:[0,+\infty[\times\Omega\rightarrow\mathbf{R}^N$
describing the portfolio holdings. The stochastic process $(x^
j_t)_{t\in[0,+\infty[}$ represents the number of pieces of $j$th
asset portfolio held by the portfolio as time goes by. Remark that
the It\^{o} stochastic integral
\begin{equation}
\int_0^tx\cdot dS=\int_0^tx_u\cdot dS_u,
\end{equation}
\noindent and the Stratonovich's stochastic integral
\begin{equation}\label{strat}
\int_0^tx\circ dS:=\int_0^tx\cdot d
S+\frac{1}{2}\int_0^td\left<x,S\right>=\int_0^tx_u\cdot d
S_u+\frac{1}{2}\int_0^td \left<x,S\right>_u
\end{equation}
 are well defined
for this choice of integrator ($S$) and integrand ($x$), as long as
the strategy is \textbf{admissible}. We mean by this that $x$  is a
predictable semimartingale for which the It\^{o} integral $\int_0^tx\cdot dS\ge-v$ is a.s.
for some $v>0$ and all $t$. Thereby, the bracket $\left<\cdot,
\cdot\right>$ denotes the continuous part of the quadratic covariation of two processes. In a general context
strategies do not need to be semimartingales, but if we want the quadratic covariation in (\ref{strat}) and hence
 Stratonovich's integral to be well defined, we must require this additional assumption.
For details about stochastic integration we refer to Appendix A in
\cite{Em89}, which summarizes Chapter VII of the authoritative
\cite{DeMe80}. The portfolio value is the process
$(V_t)_{t\in[0,+\infty[}$ defined by
\begin{equation}
V_t:=V_t^x:=x_t\cdot S_t.
\end{equation}
An admissible strategy $x$ is said to be \textbf{self-financing}
if and only if the portfolio value at time $t$ is given by
\begin{equation}
V_t=V_0+\int_0^tx_u\cdot dS_u.
\end{equation}
\noindent This means that the portfolio gain is  the It\^{o}
integral of the strategy with the price process as integrator: the
change of portfolio value is purely due to changes of the assets'
values. The self-financing condition can be rewritten in
differential form as
\begin{equation}
dV_t=x_t\cdot dS_t.
\end{equation}
As pointed out in \cite{BjHu05}, if we want to utilize
Stratonovich's integral to rephrase the self-financing condition,
while maintaining its economical interpretation (which is necessary
for the subsequent constructions of mathematical finance), we write
\begin{equation}
V_t=V_0+\int_0^tx_u\circ dS_u-\frac{1}{2}\int_0
^td\left<x,S\right>_u
\end{equation}
or, equivalently
\begin{equation}
dV_t=x_t\circ dS_t-\frac{1}{2}\,d\left<x,S\right>_t.
\end{equation}
\par An \textbf{arbitrage strategy} (or arbitrage for short) for the market model is an admissible self-financing
strategy $x$, for which one of the following condition holds for
some horizon $T>0$:
\begin{itemize}
\item $P[V_0^{x}<0]=1$ and $P[V_T^{x}\ge0]=1$,
\item $P[V_0^{x}\le0]=1$ and $P[V_T^{x}\ge0]=1$ with $P[V_T^{x}>0]>0$.
\end{itemize}
In Chapter 9 of \cite{DeSc08} the no arbitrage condition is given
a topological characterization. In view of the fundamental Theorem
of asset pricing, the no-arbitrage condition is substituted by a
stronger condition, the so called
no-free-lunch-with-vanishing-risk.
\begin{definition}\label{def1}
Let $(S_t)_{t\in[0,+\infty[}$ be a semimartingale and
$(x_t)_{t\in[0,+\infty[}$ and admissible strategy. We denote by
$(x\cdot S)_{+\infty}:=\lim_{t\rightarrow +\infty}\int_0^tx_u\cdot
dS_u$, if such limit exists, and by $K_0$ the subset of
$L^0(\Omega,\mathcal{A}_{\infty},P)$ containing all such $(x\cdot
S)_{+\infty}$. Then, we define
\begin{itemize}
\item $C_0:=K_0 - L^0_{+}(\Omega,\mathcal{A}_{\infty},\mathbb{P})$.
\item $C:=C_0\cap L^{\infty}_{+}(\Omega,\mathcal{A}_{\infty},\mathbb{P})$.
\item $\bar{C}$: the closure of $C$ in $L^{\infty}$ with respect to
the norm topology.
\end{itemize}
The market model satisfies
\begin{itemize}
\item the \textbf{no arbitrage condition (NA)} if and only if
$C\cap L^{\infty}(\Omega,\mathcal{A}_{\infty},\mathbb{P})=\{0\}$,
and
\item the \textbf{ no-free-lunch-with-vanishing-risk condition (NFLVR)} if and only if
$\bar{C}\cap
L^{\infty}(\Omega,\mathcal{A}_{\infty},\mathbb{P})=\{0\}$.
\end{itemize}
\end{definition}

\noindent Delbaen and Schachermayer proved in 1994 (see
\cite{DeSc08} Chapter 9.4, in particular the main Theorem 9.1.1)
\begin{theorem}[\textbf{Fundamental Theorem of Asset Pricing in Continuous
Time}]\label{ThmDeSch} Let $(S_t)_{t\in[0,+\infty[}$ and
$(\hat{S}_t)_{t\in[0,+\infty[}$ be  bounded semimartingales. There
is an equivalent martingale measure $\mathbb{P}^*$ for the
discounted prices $\hat{S}$ if and only if the market model
satisfies the (NFLVR).
\end{theorem}
This is a generalization for continuous time of the
Dalang-Morton-Willinger Theorem proved in 1990 (see \cite{DeSc08},
Chapter 6) for the discrete time case, where the (NFLVR) is relaxed
to the (NA) condition. The Dalang-Morton-Willinger Theorem
generalizes to arbitrary probability spaces the Harrison and Pliska
Theorem (see \cite{DeSc08}, Chapter 2) which holds true in discrete
time for finite probability spaces.\par An equivalent alternative to
the martingale measure approach for asset pricing purposes is given
by the pricing kernel (state price deflator) method.
\begin{definition}
Let $(S_t)_{t\in[0,+\infty[}$  be a semimartingale describing the
price process for the assets of our market model. The positive
semimartingale $(\beta_t)_{t\in[0,+\infty[}$ is called
\textbf{pricing kernel (or state price deflator)} for $S$ if and only if
$(\beta_tS_t)_{t\in[0,+\infty[}$ is a $\mathbb{P}$-martingale.
\end{definition}

\begin{theorem}\label{ThmZ}
Let $(S_t)_{t\in[0,+\infty[}$ and $(\hat{S}_t)_{t\in[0,+\infty[}$
be  bounded semimartingales. The process $\hat{S}$ admits an
equivalent martingale measure $\mathbb{P}^*$ if and only if there
is a pricing kernel $\beta$ for $S$ (or for $\hat{S}$), which is a $\mathbb{P}$-martingale.
\end{theorem}

As shown in \cite{HuKe04} (Chapter 7, definitions 7.18, 7.47 and
Theorem 7.48), the existence of a pricing kernel is equivalent to
the existence of an equivalent martingale measure for a specific choice of num\'{e}raire.
If we want the num\'{e}raire to be arbitrary, like the one we originally choose for the model, then we have to additionally assume that the pricing kernel $\beta$ is a $\mathbb{P}$-martingale.\par

In economic theory the value of an investment is given by the
present value of its future cashflows. This idea can be
mathematically formalized in terms of the market model presented
so far by introducing the following
\begin{definition}[\textbf{Cashflows and Intensities}] Let $(S_t)_{t\in[0,+\infty[}$ be the $\mathbf{R}^N$
valued semimartingale representing nominal prices,
given a certain num\'{e}raire with value process
$(B_t)_{t\in[0,+\infty[}$. All process are adapted to the
filtration $\mathcal{A}$. The asset \textbf{stochastic cashflow intensities}
are given by the semimartingale $(c_t)_{t\in[0,+\infty[}$ defined
as
\begin{equation}
c_t:=-\lim_{h\rightarrow 0^+}\mathbb{E}_t\left[\frac{S_{t+h}
-S_{t}}{h}\right]+r^0_tS_t,
\end{equation}
wherever the limit is defined.
\noindent The components of a vector valued process
$(C_t)_{t\in[0,+\infty[}$ satisfying the It\^{o} integral equation
\begin{equation}
C_t=\int_{t^{-}}^{t^+}dc_h
\end{equation}
are termed  \textbf{stochastic cashflows}.
\end{definition}
\noindent For example, a bond is identified with its future
coupons and its nominal, and a stock is identified with all its
future dividends. In the (straight) bond case the cashflow
 is deterministic, has discontinuities at the coupon
payment dates and vanishes after maturity. In the stock case the
cashflow  is stochastic, has discontinuities at the dividend
payment dates and has an unbounded support. In these two cases
intensities exist as stochastic generalized functions.
\begin{theorem}\label{ThmNPV}
Let $(S_t)_{t\in[0,+\infty[}$ and $(c_t)_{t\in[0,+\infty[}$ be
 bounded semimartingales, and the cash account $j=0$ be the num\'{e}raire. If the market model satisfies the
NFLVR condition, then
\begin{equation}
S_t=\mathbb{E}^*_t\left[\int_t^{+\infty}dh\,c_h
\exp\left(-\int_t^hdu\,r^0_u\right)\right]=\frac{1}{\beta_t}\mathbb{E}_t\left[\int_t^{+\infty}dh\,c_h\beta_{h}\right],
\end{equation}
where $\mathbb{E}_t^*$ denotes
the risk neutral conditional expectation, and the martingale
$\beta$ the state price deflator.
\end{theorem}

\subsection{Geometric Reformulation of the Market Model: Primitives}
We are going to introduce a more general representation of the
market model introduced in section \ref{StochasticPrelude}, which
better suits to the arbitrage modelling task. In this subsection
we extend the terminology introduced by \cite{SmSp98} for the time
discrete case to the generic one.
\begin{defi}\label{defi1}
A \textbf{gauge} is an ordered pair of two $\mathcal{A}$-adapted
real valued semimartingales $(D, P)$, where
$D=(D_t)_{t\ge0}:[0,+\infty[\times\Omega\rightarrow\mathbf{R}$ is
called \textbf{deflator} and
$P=(P_{t,s})_{t,s}:\mathcal{T}\times\Omega\rightarrow\mathbf{R}$,
which is called \textbf{term structure}, is considered as a
stochastic process with respect to the time $t$, termed
\textbf{valuation date} and
$\mathcal{T}:=\{(t,s)\in[0,+\infty[^2\,|\,s\ge t\}$. The parameter
$s\ge t$ is referred as \textbf{maturity date}. The following
properties must be satisfied a.s. for all $t, s$ such that $s\ge
t\ge 0$:
 \begin{itemize}
  \item [(i)] $P_{t,s}>0$,
  \item [(ii)] $P_{t,t}=1$.
 \end{itemize}
\end{defi}

\begin{remark}
Deflators and term structures can be considered
\textit{outside the context of fixed income.} An arbitrary
financial instrument is mapped to a gauge $(D, P)$ with the
following economic interpretation:
\begin{itemize}
\item Deflator: $D_t$ is the value of the financial instrument at time $t$ expressed in terms of some num\'{e}raire. If we
choose the cash account, the $0$-th asset as num\'{e}raire, then
we can set $D_t^j:=\hat{S}_t^j=\frac{S_t^j}{S_t^0}\quad(j=1,\dots
N)$.
\item Term structure: $P_{t,s}$ is the value at time $t$ (expressed in units of
deflator at time $t$) of a synthetic zero coupon bond with
maturity $s$ delivering one unit of financial instrument at time
$s$. It represents a term structure of forward prices with respect
to the chosen num\'{e}raire.
\end{itemize}
\noindent We point out that there is no unique choice for
deflators and term structures describing an asset model. For
example, if a set of deflators qualifies, then we can multiply
every deflator  by the same positive semimartingale to obtain
another suitable set of deflators. Of course term structures have
to be modified accordingly. The term "deflator" is clearly
inspired by actuarial mathematics. In the present context it
refers to a nominal asset value up division by a strictly positive
semimartingale (which can be the state price deflator if this
exists and it is made to the num\'{e}raire). There is no need to
assume that a deflator is a positive process. However, if we want
to make an asset to our num\'{e}raire, then we have to make sure
that the corresponding deflator is a strictly positive stochastic
process.
\end{remark}
\begin{example}{Stock Index}\\
Let us consider a total return stock index, where the dividends are
reinvested.
\begin{itemize}
\item $D_t=$ stock index value at time $t$ expressed in terms of the cash asset (risk free discounting).
\item $P_{t,s}=$ price of a forward on the stock index issued at time
$t$ maturing at time $s$ expressed in terms of $D_t$.
\end{itemize}
\end{example}

\begin{example}{Zero Bonds}\\
Let us consider a family of maturing zero bonds.
\begin{itemize}
\item $D_t\equiv1=$ value of a zero bond maturing at time $t$ = value of one unit of cash at time $t$ expressed in terms of the cash asset itself.
\item $P_{t,s}=$ price of a zero bond issued at time $t$ and delivering one unit of cash at time $s$ expressed in terms of $D_t$.
\end{itemize}
\end{example}

\noindent Deflators typically represent for a currency the time
evolution of inflation or deflation. Quotients of deflators are
exchange rates.
\begin{example}{Exchange Rates}
\begin{equation}\frac{D_t^{USD}}{D_t^{CHF}}=FX_t^{CHF \rightarrow USD}.\end{equation}
\end{example}

\subsection{Geometric Reformulation of the Market Model: Portfolios}\label{trans}
We want now to introduce transforms of deflators and term structures
in order to group gauges containing the same (or less) stochastic
information. That for, we will consider \textit{deterministic}
linear combinations of assets modelled by the same gauge (e. g. zero
bonds of the same credit quality with different maturities).

\begin{defi}\label{gaugeTransforms2}
Let $\pi:[0, +\infty[\longrightarrow \mathbf{R}$ be a deterministic
cashflow intensity (possibly generalized) function. It induces a
\textbf{gauge transform} $(D,P)\mapsto
\pi(D,P):=(D,P)^{\pi}:=(D^{\pi}, P^{\pi})$ by the formulae
\begin{equation}
\boxed{ D_t^{\pi}:=D_t\int_0^{+\infty}dh\,\pi_h P_{t, t+h}\qquad
P_{t,s}^{\pi}:=\frac{\int_0^{+\infty}dh\,\pi_h P_{t,
s+h}}{\int_0^{+\infty}dh\,\pi_h P_{t, t+h}}.}
\end{equation}

\end{defi}

\begin{remark} The cashflow intensity $\pi$ specifies the bond cashflow
structure. The bond value at time $t$ expressed in terms of the
market model  num\'{e}raire is given by $D_t^{\pi}$.  The term
structure of forward prices for the bond future expressed in terms
of the bond current value is given by $P_{t,s}^{\pi}$. \par A
gauge transform is well defined if and only if the integrals are
convergent, which is the case if
$\limsup_{h\rightarrow\infty}\exp\left(\frac{\log(|\pi_h|)}{h}\right)\le
1$. A gauge transform with positive cashflows always maps a gauge
to another gauge. A generic gauge transform does not, since the
positivity of term structures is not a priori preserved.
Therefore, when referring to a generic gauge transform, it is
necessary to specify its domain of definition, that is the set of
gauges which are mapped to other gauges.
\end{remark}
We can use gauge
transforms to construct portfolios of instruments already modelled
by a known gauge.

\begin{example}{Coupon Bonds}\\
Let $(D,P)$ the gauge describing the family of zero bonds. To model
a family of straight coupon bonds with coupon rate $g$ and term to
maturity $T$ let us choose the (generalized) cashflow intensity
function
\begin{equation}
\pi_t:=\sum_{s=1}^{T-1}g\delta_{t-s}+(1+g)\delta_{t-T}.
\end{equation}
\noindent Thereby $\delta$ denotes the Dirac-delta generalized
function.
\begin{itemize}
\item $D_t^{\pi}=$ value at time $t$ of a coupon bond issued at time $t$.
\item $P_{ts}^{\pi}=$ price of a synthetic zero bond issued at time $t$ and delivering at time $s$ a coupon bond (issued at time $s$), expressed in terms of $D_t^{\pi}$.
\end{itemize}
\end{example}

\begin{proposition}\label{conv}
Gauge transforms induced by cashflow vectors have the following
property:
\begin{equation}((D,P)^{\pi})^{\nu}= ((D,P)^{\nu})^{\pi} = (D,P)^{\pi\ast\nu},\end{equation} where
$\ast$ denotes the convolution product of two cashflow vectors or
intensities respectively:
\begin{equation}\label{convdef}
    (\pi\ast\nu)_t:=\int_0^tdh\,\pi_h\nu_{t-h}.
\end{equation}

\end{proposition}

The convolution of two non-invertible gauge transform is
non-invertible. The convolution of a non-invertible with an
invertible gauge transform is non-invertible.

\begin{defi}
An invertible gauge transform is called \textbf{non-singular}. Two
gauges are said to be in \textbf{same orbit} if and only if there is
a non-singular gauge transform mapping one onto the other. A
singular gauge transform $\pi$ defines a partial ordering $(D,
P)\succ (D^\pi, P^\pi)$ in the set of gauges. $(D, P)$ is said to be
in a \textbf{higher orbit} than $(D^\pi, P^\pi)$.
\end{defi}

It is therefore possible to construct gauges in a lower orbit from
higher orbits, but not the other way around. Orbits represent assets
containing equivalent information. For every orbit it suffices
therefore to specify only one gauge.\\\\

\begin{figure}[h!]
\centering
\includegraphics[width = 12.5cm, angle=-90]{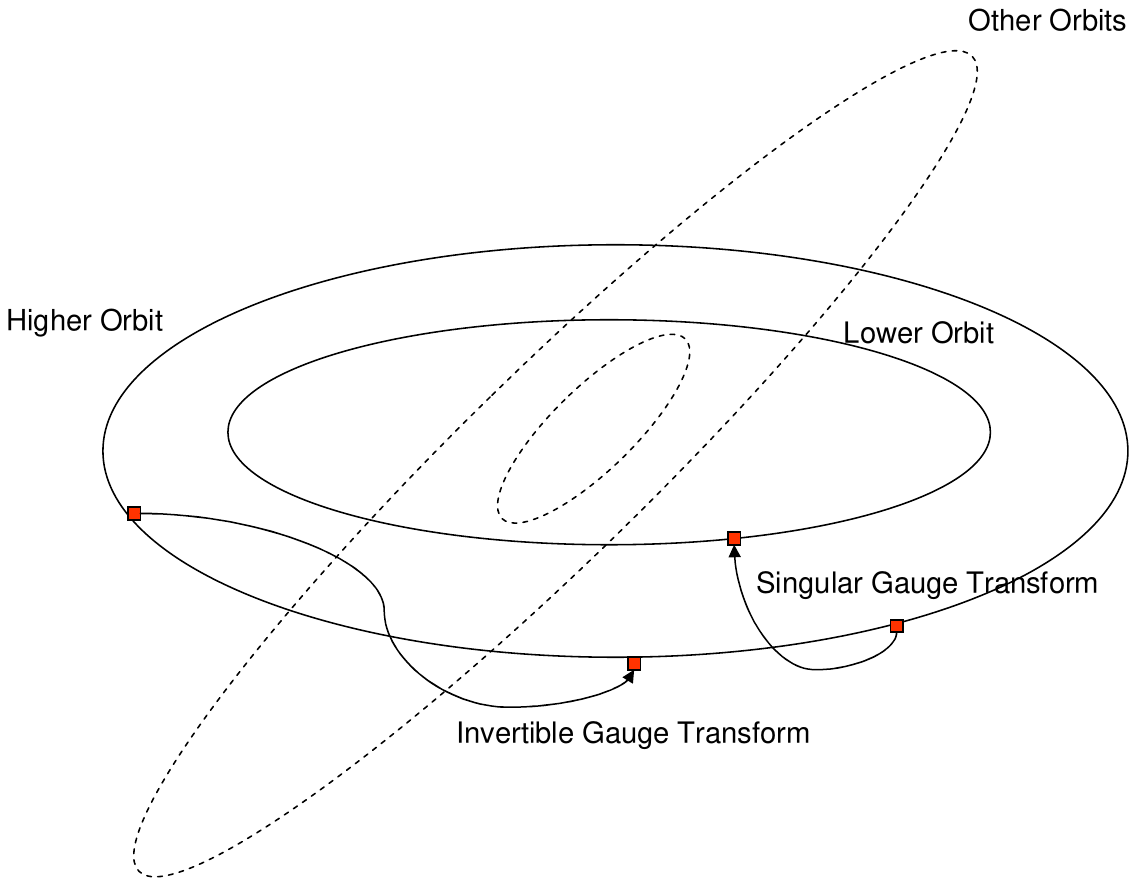}
\caption{Gauge Transforms}\label{MM1}
\end{figure}

\begin{defi}\label{int}
A gauge $(D, P)$ with term structure $P=(P_{t,s})_{t,s}$ satisfies
the \textbf{positive interest condition} if and only if for all
$t$ the function $s\mapsto P_{ts}$ is strictly monotone
decreasing. Such a gauge is said to be \textbf{positive}. A gauge
not satisfying this property is termed \textbf{principal gauge}.
The term structure can be written as a functional of the
\textbf{instantaneous forward rate } f defined as
\begin{equation}
\boxed{
  f_{t,s}:=-\frac{\partial}{\partial s}\log P_{t,s},\quad
  P_{t,s}=\exp\left(-\int_t^sdhf_{t,h}\right). }
\end{equation}
\noindent and
\begin{equation}
\boxed{
 r_t:=\lim_{s\rightarrow t^+}f_{t,s}
 }
\end{equation}
\noindent is termed \textbf{short rate}.
\end{defi}

\begin{remark}
Since $(P_{t,s})_{t,s}$ is a $t$-stochastic process (semimartingale) depending on a parameter $s\ge t$, the $s$-derivative can be defined deterministically, and the expressions above make sense pathwise in a both classical and generalized sense. In a generalized sense we will always have a $\mathcal{D}^{\prime}$ derivative for any $\omega\in \Omega$; this corresponds to a classic $s$-continuous  derivative if $P_{t,s}(\omega)$ is a $C^1$-function of $s$ for any fixed $t\ge0$ and $\omega\in\Omega$.
\end{remark}

We see that the positive interest condition is satisfied if and only
if $f_{t,s}>0$ for all $t,s$, $s\ge t$. The positive interest
condition is associated with the \textit{storage requirement}.
Whenever it is always more valuable to get a piece of a financial
object today than in the future, then it should be modelled with a
gauge satisfying the positive interest condition. Examples are: non
perishable goods, currencies, price indices for equities and real
estates, total return indices. Examples of financial quantities not
satisfying the positive interest condition and thus reflecting items
which are not storable, are: inflation indices, short rates,
dividend indices for equities, rental indices for real estates.

\begin{defi}
The cash flow intensity $[-1]:=\delta^{\prime}$, first derivative
of the Dirac delta generalized function, defines the \textbf{short
rate transform},
\begin{equation}
\begin{split}
D_t^{[-1]}&=D_t\int_0^{+\infty}dh\,\delta^{\prime}_h P_{t, t+h}=D_tr_t\\
P_{t,s}^{[-1]}&=\frac{\int_0^{+\infty}dh\,\delta^{\prime}_h P_{t,
s+h}}{\int_0^{+\infty}dh\,\delta^{\prime}_h P_{t,
t+h}}=\frac{f_{t,s}}{r_t}P_{t,s}
\end{split}
\end{equation}
while the cash flow intensity $[+1]:=\Theta$, Heavyside function,
defines the \textbf{perpetuity transform}
\begin{equation}
\begin{split}
D_t^{[+1]}&=D_t\int_0^{+\infty}dh\,P_{t,t+h}\\
P_{t,s}^{[+1]}&=\frac{\int_0^{+\infty}dh\,P_{t,
s+h}}{\int_0^{+\infty}dh\,P_{t, t+h}}.
\end{split}
\end{equation}
\end{defi}
\begin{figure}[h!]
\centering
\includegraphics[width = 12.5cm, angle=-90]{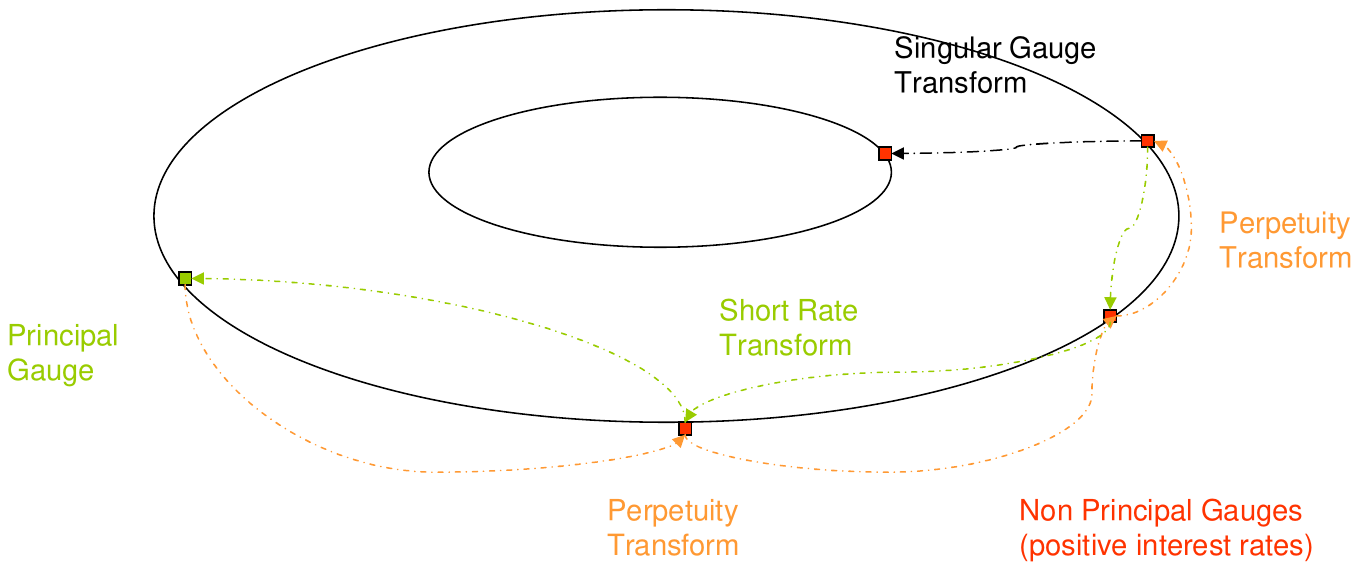}
\caption{Short Rate and Perpetuity Transforms}\label{MM2}
\end{figure}

\begin{notation}
Repeated application of perpetuity and short rate transforms are
given by:
\begin{equation}
\begin{split}
[0]:&=\delta\text{: Dirac delta generalized function}\\
[+1]:&=\Theta\text{: Heavyside function}\\
[+k]_t:&=\frac{t^{k-1}}{(k-1)!}\quad(k\ge2)\\
[-1]_t:&=\delta^{\prime}\text{: first derivative of Dirac delta}\\
[-k]_t:&=\delta^{(k)}\text{: k-th derivative of Dirac
delta,}\,(k\ge2)
\end{split}
\end{equation}
 Thereby, for any integers $m,n$ one
has $[k]\ast[l]=[k+l]$ (cf. \cite{Ho03} Chapter IV).
\end{notation}

The short rate and the perpetuity transform are inverse to another,
as one can see from Proposition \ref{conv} and $[+1]\ast[-1]=[0]$.
The short rate transform can be applied only to a positive gauge
producing a gauge which possibly does not satisfy the positive
interest rate condition. The perpetuity transform is a gauge
transform that can be applied to any gauge producing always a
positive gauge.

\begin{proposition}
A gauge satisfies the positive interest condition if and only if it
can be obtained as the perpetuity transform of some other gauge.
\end{proposition}

\noindent The positive interest condition is difficult to satisfy
for a stochastic model of a gauge.

\begin{example}{Fixed Income, Equity and Real Estate Gauges}\\
\begin{figure}[h!]
\centering
\includegraphics[height = 12.5cm, angle=-90]{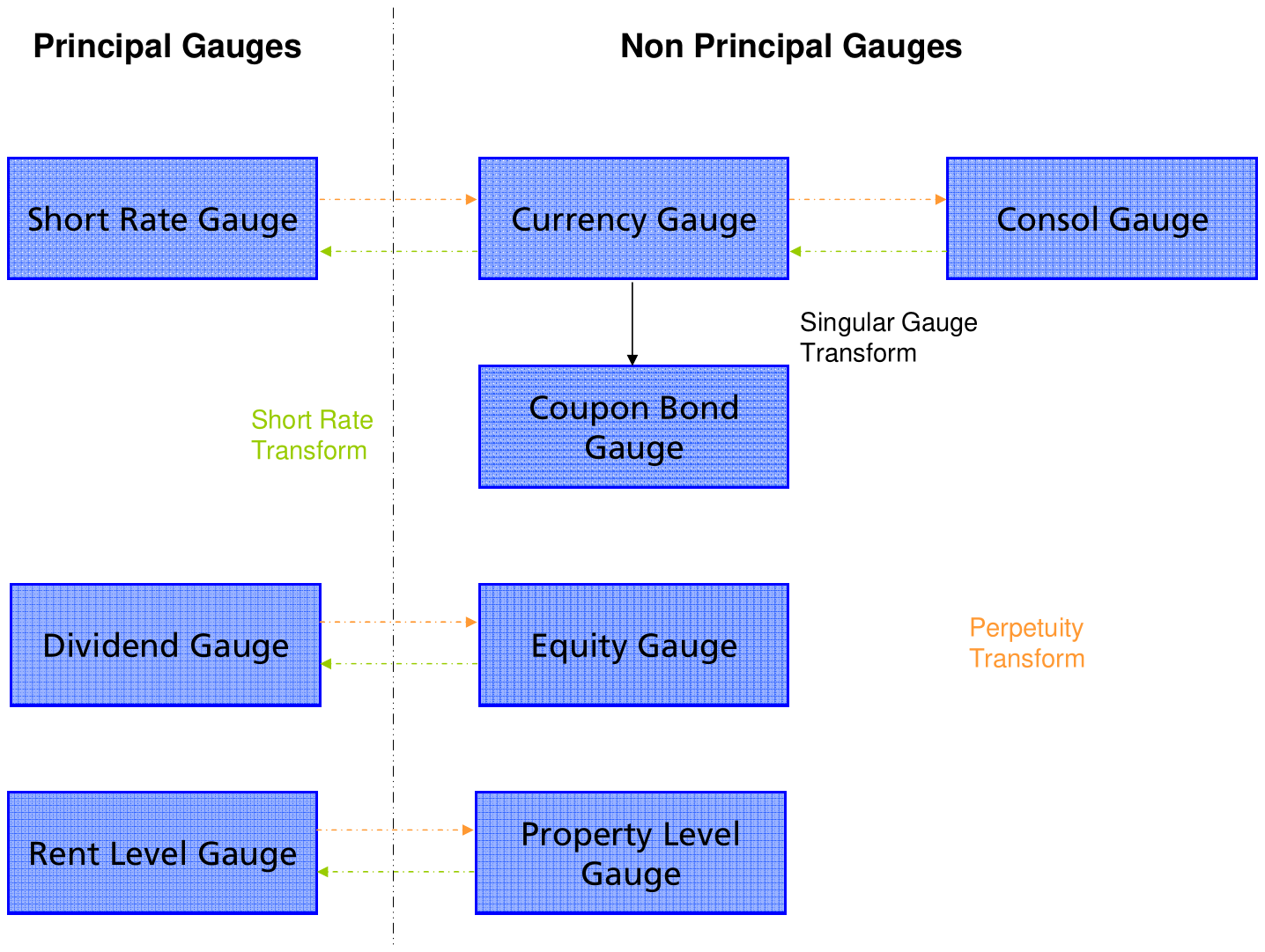}
\caption{Gauges}\label{FIERE}
\end{figure}
\end{example}

\begin{remark} The special choice of vanishing interest rate $r\equiv0$ or flat term structure
$P\equiv1$ for all assets corresponds to the classical model,
where only asset prices and their dynamics are relevant. We will
analyze this case in detail in the guiding example presented in
section \ref{GuidEx}.
\end{remark}

\section{Arbitrage Theory in a Differential Geometric
Framework}\label{foundations} Now we are in the position to
rephrase the asset model presented in subsection
\ref{StochasticPrelude} in terms of a natural geometric language.
That for, we will unify Smith's and Ilinski's ideas to model a
simple market of $N$ base assets. In Smith and Speed
(\cite{SmSp98}) there is no explicit differential geometric
modelling but the use of an allusive terminology (e.g. gauges,
gauge transforms). In Ilinski (\cite{Il01}) there is a
construction of a principal fibre bundle allowing to express
arbitrage in terms of curvature. Our construction of the principal
fibre bundle will differ from Ilinski's one in the choice of the
group action and the bundle covering the base space. Our choice
encodes Smith's intuition in differential geometric language. \par
In this paper we explicitly model no derivatives of the base
assets, that is, if derivative products have to be considered,
then they have to be added to the set of base assets. The
treatment of derivatives of base assets is tackled in
(\cite{FaVa12}).  Given $N$ base assets we want to construct a
portfolio theory and study arbitrage. Since arbitrage is
explicitly allowed, we cannot a priori assume the existence of a
risk neutral measure or of a state price deflator. In terms of
differential geometry, we will adopt the mathematician's and not
the physicist's approach. The market model is seen as a principal
fibre bundle of the (deflator, term structure) pairs, discounting
and foreign exchange as a parallel transport, num\'{e}raire as
global section of the gauge bundle, arbitrage as curvature. The
Ambrose-Singer Theorem allows to parameterize arbitrage strategies
as element of the Lie algebra of the holonomy group. The
no-free-lunch-with-vanishing-risk condition is proved to be
equivalent to a zero curvature condition or to a continuity
equation allowing for an hydrodynamics study of arbitrage flows.

\subsection{Market Model as Principal Fibre Bundle} As a concise
general reference for principle fibre bundles we refer to
Bleecker's book (\cite{Bl81}). More extensive treatments can be
found in Dubrovin, Fomenko and Novikov (\cite{DuFoNo84}), and in
the classical Kobayashi and Nomizu (\cite{KoNo96}). Let us
consider -in continuous time- a market with $N$ assets and a
num\'{e}raire. A general portfolio at time $t$ is described by the
vector of nominals $x\in X$, for an open set
$X\subset\mathbf{R}^N$. Following Definition \ref{defi1}, the asset model induces for $j=1,\dots,N$ the gauge
\begin{equation}(D^j,P^j)=((D_t^j)_{t\in[0, +\infty[},(P_{t,s}^j)_{s\ge t}),\end{equation}
\noindent where $D^j$ denotes the deflator and $P^j$ the term
structure. This can be written as
\begin{equation}P_{t,s}^j=\exp\left(-\int_t^sf^j_{t,u}du\right),\end{equation}
where $f^j$ is the instantaneous forward rate process for the $j$-th asset and the corresponding short rate is given by $r_t^j:=\lim_{u\rightarrow 0^+}f^j_{t,u}$. For a
portfolio with nominals $x\in X\subset\mathbf{R}^N$ we define
\begin{equation}
\boxed{ D_t^x:=\sum_{j=1}^Nx_jD_t^j\quad
f_{t,u}^x:=\sum_{j=1}^N\frac{x_jD_t^j}{\sum_{j=1}^Nx_jD_t^j}f_{t,u}^j\quad
P_{t,s}^x:=\exp\left(-\int_t^sf^x_{t,u}du\right).}
\end{equation}
The short rate writes
\begin{equation}
r_t^x:=\lim_{u\rightarrow t^+}f^x_{t,u}=\sum_{j=1}^N\frac{x_jD_t^j}{\sum_{j=1}^Nx_jD_t^j}r_t^j.
\end{equation}

\noindent The image space of all possible strategies reads
\begin{equation}\boxed{M:=\{(x,t)\in X\times[0, +\infty[\}.}\end{equation}
In subsection \ref{trans} cashflow intensities and the corresponding
gauge transforms were introduced. They have the structure of an
Abelian semigroup
\begin{equation}
\boxed{
 G:=\mathcal{E}^{\prime}([0,
+\infty[,\mathbf{R})=\{F\in\mathcal{D}^{\prime}([0,+\infty[)\mid
\text{supp}(F)\subset[0, +\infty[\text{ is compact}\},
}
\end{equation}
where the semigroup operation on distributions with compact support
is the convolution (see \cite{Ho03}, Chapter IV), which extends the
convolution of regular functions as defined by formula
(\ref{convdef}).
\begin{defi}
The \textbf{Market Fibre Bundle} is defined as the fibre bundle of
gauges
\begin{equation}
\boxed{ \mathcal{B}:=\{({{D^\pi}_t}^x,{{P^\pi}_{t,\,\cdot}}^x)|(x,t)\in
M, \pi\in G^*\}. }
\end{equation}
\end{defi}\noindent
The cashflow intensities defining invertible transforms constitute
an Abelian group
\begin{equation}
\boxed{ G^*:=\{\pi\in G|\text{ it exists } \nu\in G\text{ such that
}\pi\ast\nu=[0]\}\subset \mathcal{E}^{\prime}([0,
+\infty[,\mathbf{R}).}\end{equation} From Proposition \ref{conv} we
obtain
\begin{theorem} The market fibre bundle $\mathcal{B}$ has the
structure of a $G^*$-principal fibre bundle  given by the action
\begin{equation}
\begin{split}
\mathcal{B}\times G^* &\longrightarrow\mathcal{B}\\
 ((D,P), \pi)&\mapsto (D,P)^{\pi}=(D^{\pi},P^{\pi}).
\end{split}
\end{equation}
\noindent The group $G^*$ acts freely and differentiably on
$\mathcal{B}$ to the right.
\end{theorem}

\subsection{Num\'{e}raire as Global Section of the Bundle of
Gauges} If we want to make an arbitrary  portfolio of the given
assets specified by the nominal vector $x^{\text{Num}}$ to our
num\'{e}raire, we have to renormalize all deflators by an
appropriate gauge transform $\pi^{\text{Num}, x}$ so that:
\begin{itemize}
\item The portfolio value is constantly over time normalized to one:
  \begin{equation}D_t^{x^{\text{Num}},\pi^{\text{Num}}}\equiv 1.\end{equation}
\item All other assets' and portfolios' are expressed in terms of
the num\'{e}raire:
\begin{equation}D_t^{x,\pi^{\text{Num}}}=\text{FX}_t^{x\rightarrow x^\text{Num}}:=\frac{D^x_t}{D^{x^\text{Num}}_t}.\end{equation}
\end{itemize}
It is easily seen that the appropriate choice for the gauge
transform $\pi^{\text{Num}}$ making the portfolio $x^{\text{Num}}$
to the num\'{e}raire is given by the global section of the bundle
of gauges defined by
\begin{equation}\boxed{\pi^{\text{Num},x}_t:=\text{FX}_t^{x\rightarrow x^{\text{Num}}}.}\end{equation}
\noindent Of course such a gauge transform is well defined if and
only if the num\'{e}raire deflator is a positive semimartingale.

\subsection{Cashflows as Sections of the Associated Vector
Bundle} By choosing the fiber $V:=\mathbf{R}^{[0, +\infty[}$
and the representation $\rho:G\rightarrow \text{GL}(V)$ induced by
the gauge transform definition, and therefore satisfying the
homomorphism relation $\rho(g_1\ast g_2)=\rho(g_1)\rho(g_2)$, we
obtain the associated vector bundle $\mathcal{V}$. Its sections
represents cashflow streams - expressed in terms of the deflators
- generated by portfolios of the base assets. If
$v=(v^x_t)_{(x,t)\in M}$ is the \textit{deterministic} cashflow
stream, then its value at time $t$ is equal to
\begin{itemize}
\item the deterministic quantity $v_t^x$, if the value is measured in terms of
the deflator $D_t^x$,
\item the stochastic quantity $v^x_tD^x_t$, if the value is measured in terms of the
num\'{e}raire (e.g. the cash account for the choice
$D_t^j:=\hat{S}_t^j$ for all $j=1,\dots,N$).
\end{itemize}
 \noindent In the general theory of
principal fibre bundles, gauge transforms are bundle automorphisms
preserving the group action and equal to the identity on the base
space. Gauge transforms of $\mathcal{B}$ are naturally isomorphic to
the sections of the bundle $\mathcal{B}$ (See Theorem 3.2.2 in
\cite{Bl81}). Since $G^*$ is Abelian, right multiplications are
gauge transforms. Hence, there is a bijective correspondence between
gauge transforms and cashflow intensities admitting an inverse. This
justifies the terminology introduced in Definition
\ref{gaugeTransforms2}.

\subsection{Derivatives of Stochastic
Processes}\label{Derivatives} One of the main contribution of this
paper is to reformulate stochastic finance in a natural geometric
language. In stochastic differential geometry one would like to lift
the constructions of stochastic analysis from open subsets of
$\mathbf{R}^N$ to  $N$ dimensional differentiable manifolds. To that
aim, chart invariant definitions are needed and hence a stochastic
calculus satisfying the usual chain rule and not It\^{o}'s Lemma is required.
(cf. \cite{HaTh94}, Chapter 7, and the remark in Chapter 4 at the
beginning of page 200). That is why we will be mainly concerned in
the following by stochastic integrals and derivatives meant in the
sense of \textit{Stratonovich} and not of \textit{It\^{o}}. Following \cite{Gl11} and \cite{CrDa07} we introduce the following
\begin{defi}\label{Nelson}
Let $I$ be a real interval and $Q=(Q_t)_{t\in I}$ be a  vector valued stochastic process on the probability space
$(\Omega, \mathcal{A}, P)$. The process $Q$ determines three families of $\sigma$-subalgebras of the $\sigma$-algebra $\mathcal{A}$:
\begin{itemize}
\item[(i)] ''Past'' $\mathcal{P}_t$, generated by the preimages of Borel sets in $\mathbf{R}^N$  by all mappings $Q_s:\Omega\rightarrow\mathbf{R}^N$ for $0<s<t$.
\item[(ii)] ''Future'' $\mathcal{F}_t$, generated by the preimages of Borel sets in $\mathbf{R}^N$  by all mappings $Q_s:\Omega\rightarrow\mathbf{R}^N$ for $0<t<s$.
\item[(iii)] ''Present'' $\mathcal{N}_t$, generated by the preimages of Borel sets in $\mathbf{R}^N$  by the mapping $Q_t:\Omega\rightarrow\mathbf{R}^N$.
\end{itemize}
Let $Q=(Q_t)_{t\in I}$ be a $\mathcal{S}^0(I)$-process, i.e. a process with continuous sample paths and adapted to $\mathcal{P}$ and $\mathcal{F}$, so that $t\mapsto Q_t$ is a continuous mapping continuous mappings from $I$ to $L^2(\Omega, \mathcal{A})$.
 Assuming that the following limits exist,
\textbf{Nelson's stochastic derivatives} are defined as
\begin{equation}
\boxed{
\begin{split}
&\mathfrak{D}Q_t:=\lim_{h\rightarrow
0^+}\mathbb{E}\Br{\left.\frac{Q_{t+h}-Q_t}{h}\right|
\mathcal{P}_t}\text{: forward derivative,}\\
& \mathfrak{D}_*Q_t:=\lim_{h\rightarrow
0^+}\mathbb{E}\Br{\left.\frac{Q_{t}-Q_{t-h}}{h}\right|
\mathcal{F}_{t}}\text{: backward derivative,}\\
&\mathcal{D}Q_t:=\frac{\mathfrak{D}Q_t+\mathfrak{D}_*Q_t}{2}\text{: mean derivative}.
\end{split}
}
\end{equation}
Let $\mathcal{S}^1(I)$ the set of all $\mathcal{S}^0(I)$-processes $Q$ such that
 $t\mapsto \mathfrak{D}Q_t$ and $t\mapsto \mathfrak{D}_*Q_t$ are continuous
mappings from $I$ to $L^2(\Omega, \mathcal{A})$. Let
$\mathcal{C}^1(I)$ the completion of $\mathcal{S}^1(I)$ with respect
to the norm
\begin{equation}\|Q\|:=\sup_{t\in I}\br{\|Q_t\|_{L^2(\Omega, \mathcal{A})}+\|\mathfrak{D}Q_t\|_{L^2(\Omega, \mathcal{A})}+\|\mathfrak{D}_*Q_t\|_{L^2(\Omega, \mathcal{A})}}.\end{equation}

\end{defi}

\begin{remark}
The stochastic derivatives $\mathfrak{D}$, $\mathfrak{D}_*$ and  $\mathcal{D}$
correspond to It\^{o}'s, to the anticipative and, respectively,  to Stratonovich's integral (cf. \cite{Gl11}). The process space $\mathcal{C}^1(I)$ contains all It\^{o} processes. If $Q$ is a Markov process, then the sigma algebras $\mathcal{P}_t$ (''past'') and $\mathcal{F}_t$ (''future'') in the definitions of forward and backward derivatives can be substituted by the sigma algebra $\mathcal{N}_t$ (''present''), see Chapter 6.1 and 8.1 in (\cite{Gl11}).
\end{remark}

\noindent Stochastic derivatives can be defined pointwise in $\omega\in\Omega$ outside the class $\mathcal{C}^1$ in terms of generalized functions.
\begin{defi}
Let $Q:I\times\Omega\rightarrow\mathbb{R}^N$ be a continuous linear functional in the test processes $\varphi:I\times\Omega\rightarrow\mathbb{R}^N$ for $\varphi(\cdot,\omega)\in C^{\infty}_c(I,\mathbb{R}^N)$. We mean by this that for a fixed $\omega\in\Omega$ the functional $Q(\cdot,\omega)\in\mathcal{D}(I,\mathbb{R}^N)$, the topological vector space of continuous distributions. We can then define
\textbf{Nelson's generalized stochastic derivatives:}
\begin{equation}
\boxed{
\begin{split}
&\mathfrak{D}Q(\varphi_t):=-Q(\mathfrak{D}\varphi_t)\text{: forward generalized derivative,}\\
& \mathfrak{D}_*Q(\varphi_t):=-Q(\mathfrak{D}_*\varphi_t)\text{: backward generalized derivative,}\\
&\mathcal{D}(\varphi_t):=-Q(\mathcal{D}\varphi_t)\text{: mean generalized derivative}.
\end{split}
}
\end{equation}
\end{defi}
\noindent If the generalized derivative is regular, then the process has a derivative in the classic sense. This construction is nothing else than a straightforward pathwise lift of the theory of generalized functions to a wider class stochastic processes which do not a priori allow for Nelson's derivatives in the strong sense. We will utilize this feature in the treatment of credit risk, where many processes with jumps occur.

\subsection{Stochastic Parallel Transport and Holonomy} Let us
consider the projection of $\mathcal{B}$ onto $M$
\begin{equation}
\begin{split}
p:\mathcal{B}\cong M\times G^*&\longrightarrow M\\
 (x,t,g)&\mapsto (x,t)
\end{split}
\end{equation}
and its tangential map
\begin{equation}
T_{(x,t,g)}p:\underbrace{T_{(x,t,g)}\mathcal{B}}_{\cong\mathbf{R}^N\times\mathbf{R}\times
\mathbf{R}^{[0, +\infty[}}\longrightarrow
\underbrace{T_{(x,t)}M}_{\cong\mathbf{R}^N\times\mathbf{R}.}
\end{equation}
The vertical directions are
\begin{equation}\mathcal{V}_{(x,t,g)}\mathcal{B}:=\ker\left(T_{(x,t,g)}p\right)\cong\mathbf{R}^{[0, +\infty[},\end{equation}
and the horizontal ones are
\begin{equation}\mathcal{H}_{(x,t,g)}\mathcal{B}\cong\mathbf{R}^{N+1}.\end{equation}
A connection on $\mathcal{B}$ is a projection
$T\mathcal{B}\rightarrow\mathcal{V}\mathcal{B}$. More precisely, the
vertical projection must have the form
\begin{equation}
\begin{split}
\Pi^v_{(x,t,g)}:T_{(x,t,g)}\mathcal{B}&\longrightarrow\mathcal{V}_{(x,t,g)}\mathcal{B}\\
(\delta x,\delta t, \delta g)&\mapsto (0,0,\delta g +
\Gamma(x,t,g).(\delta x, \delta t)),
\end{split}
\end{equation}
and the horizontal one must read
\begin{equation}
\begin{split}
\Pi^h_{(x,t,g)}:T_{(x,t,g)}\mathcal{B}&\longrightarrow\mathcal{H}_{(x,t,g)}\mathcal{B}\\
(\delta x,\delta t, \delta g)&\mapsto (\delta x,\delta
t,-\Gamma(x,t,g).(\delta x, \delta t)),
\end{split}
\end{equation}
such that
\begin{equation}\Pi^v+\Pi^h=\mathbf{1}_{\mathcal{B}}.\end{equation}
Stochastic parallel transport on a principal fibre bundle along a
semimartingale is a well defined construction (cf. \cite{HaTh94},
Chapter 7.4 and \cite{Hs02} Chapter 2.3 for the frame bundle case)
in terms of Stratonovich's integral. Existence and uniqueness can be
proved analogously to the deterministic case by formally
substituting the deterministic time derivative $\frac{d}{dt}$ with
the stochastic one $\mathcal{D}$ corresponding to Stratonovich's
integral.\par
 Following Ilinski's idea (\cite{Il01}), we motivate the choice
of a particular connection by the fact that it allows to encode
foreign exchange and discounting as parallel transport.
\begin{theorem}With the choice of connection
\begin{equation}\label{connection}\boxed{\Gamma(x,t,g).(\delta x, \delta t):=g\left(\frac{D_t^{\delta x}}{D_t^x}-r_t^x\delta t\right),}\end{equation}
the parallel transport in $\mathcal{B}$ has the following financial
interpretations:
\begin{itemize}
\item Parallel transport along the nominal directions ($x$-lines)
corresponds to a multiplication by an exchange rate.
\item Parallel transport along the time direction ($t$-line)
corresponds to a division by a stochastic discount factor.
\end{itemize}
\end{theorem}

Recall that time derivatives needed to define the parallel transport
along the time lines have to be understood in Stratonovich's sense.
We see that the bundle is trivial, because it has a global
trivialization, but the connection is not trivial.

\begin{proof}
Let us consider a curve $\gamma(\tau)=(x(\tau),t(\tau))$ in $M$ for
$\tau\in[\tau_1,\tau_2]$ and an element of the fiber over the
starting point $g_1\in p^{-1}(\gamma(\tau_1))\cong G$. The parallel
transport of $g_1$ along $\gamma$ is the solution $g=g(\tau)$ of the
first order differential equation
\begin{equation}
\begin{cases}
\Pi^v_{(x(\tau),t(\tau),g(\tau))}(\mathcal{D}x(\tau),\mathcal{D}t(\tau),\mathcal{D}g(\tau))=0\\
g(\tau_1)=g_1,
\end{cases}
\end{equation}
\noindent which in our case writes

\begin{equation}
\begin{cases}
\mathcal{D}g(\tau) = -g(\tau)\left(\frac{D_{t(\tau)}^{\mathcal{D}x(\tau)}}{D_{t(\tau)}^{x(\tau)}}-r_{t(\tau)}^{x(\tau)}\mathcal{D}t(\tau)\right) \\
g(\tau_1)=g_1,
\end{cases}
\end{equation}
Recall that the time derivative $\mathcal{D}$ is Nelson's derivative
corresponding to Stratonovich's integral, see subsection
\ref{Derivatives}. Now, if $\gamma$ is a nominal direction, then
$t(\tau)\equiv t$ and $\mathcal{D}t(\tau)\equiv 0$. Thus
\begin{equation}
\begin{cases}
\mathcal{D}g(\tau) = -g(\tau)\frac{\sum_{j=1}^N\mathcal{D}x_j(\tau)D_{t}^j}{\sum_{j=1}^Nx_j(\tau)D_{t}^j} \\
g(\tau_1)=g_1,
\end{cases}
\end{equation}
which means
\begin{equation}g(\tau)=g_1\frac{\sum_{j=1}^Nx_j(\tau_1)D_{t}^j}{\sum_{j=1}^Nx_j(\tau)D_{t}^j}\end{equation}
corresponding to a multiplication by an exchange rate at time $t$
from portfolio $x(\tau_1)$ to portfolio $x(\tau)$.\\
If $\gamma$ is the time direction, then $x(\tau)\equiv x$, $t(\tau)=
\tau$  and $\mathcal{D}x(\tau)\equiv 0$, $\mathcal{D}t(\tau)\equiv
1$ . Thus
\begin{equation}
\begin{cases}
\mathcal{D}g(\tau) = g(\tau)r_{\tau}^{x} \\
g(\tau_1)=g_1,
\end{cases}
\end{equation}
which means
\begin{equation}g(\tau)=g_1\exp\left(\int_{\tau_1}^{\tau}r^x_udu\right)\end{equation}
corresponding to a division by the  stochastic discount rate for
portfolio $x$
 from time $\tau_1$ to time $\tau$.
\end{proof}

\begin{remark}
Malaney and Weinstein (\cite{Ma96}) already introduced a connection
in the deterministic case in the context of self-financing basket of
goods for divisa indices. Recently, \cite{FaVa12} have elaborated a
stochastic version of the Malaney-Weinstein connection proving its
equivalence with the connection defined in (\ref{connection}).
\end{remark}

Holonomy is the group generated by the parallel transport along
closed curves. We distinguish the local from the global case.
\begin{definition}
The \textbf{holonomy group} based at $b\in \mathcal{B}$ is defined
as
\begin{equation}
\boxed{ \Hol_{b}(\chi):=\{g\in G\mid\,b \text{ and } b.g \text{
can be joined by an horizontal curve in }\mathcal{B} \}.}
\end{equation}
The \textbf{local holonomy group} based at $b\in \mathcal{B}$ is
defined as
\begin{equation}
\boxed{
\begin{split}
\Hol_{b}^0(\chi):=&\left\{g\in G\mid\,b \text{ and } b.g \text{ can be joined by a contractible horizontal}\right.\\
&\left.\text{ curve in }\mathcal{B} \right\}.
\end{split}
}
\end{equation}
\end{definition}
If $M$ and $\mathcal{B}$ are connected, then holonomy and local
holonomy depend on the base point $b$ only up to conjugation. In
this paper we will always assume connectivity for both $M$ and
$\mathcal{B}$ and therefore drop the reference to the basis point
$b$, with the understanding that the definition is good up to
conjugation.

\subsection{Nelson $\mathcal{D}$ Differentiable Market Model} We continue to
reformulate the classic asset model introduced in subsection
\ref{StochasticPrelude} in terms of stochastic differential
geometry.
\begin{definition}\label{weakMM}
 A \textbf{Nelson $\mathcal{D}$ weak differentiable market model} for
$N$ assets is described by $N$ gauges which are Nelson $\mathcal{D}$
weak differentiable with respect to the time variable. More exactly, for
all $t\in[0,+\infty[$ and $s\ge t$ there is an open time interval
$I\ni t$ such that for the deflators
$D_t:=[D_t^1,\dots,D_t^N]^{\dagger}$ and the term structures
$P_{t,s}:=[P_{t,s}^1,\dots,P_{t,s}^N]^{\dagger}$, the latter seen as
processes in $t$ and parameter $s$, there exist a $\mathcal{D}$ weak
$t$-derivative. The short rates are defined by
$r_t:=\lim_{s\rightarrow t^{+}}-\frac{\partial}{\partial s}\log
P_{ts}$.\par A strategy is a curve $\gamma:I\rightarrow X$ in the
portfolio space parameterized by the time. This means that the
allocation at time $t$ is given by the vector of nominals
$x_t:=\gamma(t)$. We denote by $\bar{\gamma}$ the lift of $\gamma$
to $M$, that is $\bar{\gamma}(t):=(\gamma(t),t)$. A strategy is said
to be \textbf{closed} if it represented by a closed curve.   A
\textbf{weak $\mathcal{D}$-admissible strategy} is predictable and
$\mathcal{D}$- weak differentiable.
\end{definition}
\noindent In general the allocation can depend on the state of the
nature i.e. $x_t=x_t(\omega)$ for $\omega\in\Omega$. Unless
otherwise specified strategies will always be weak
$\mathcal{D}$-admissible for an appropriate time interval.

\begin{proposition}
A weak $\mathcal{D}$-admissible strategy is
self-financing if and only if
\begin{equation}\label{sf}\mathcal{D}(x_t\cdot D_t)=x_t\cdot
\mathcal{D}D_t-\frac{1}{2}\mathfrak{D}_*\left<x,D\right>_t\quad\text{or}\qquad
\mathcal{D}x_t\cdot D_t=-\frac{1}{2}\mathfrak{D}_*\left<x,D\right>_t,\end{equation} almost
surely.
\end{proposition}
\begin{proof}
The strategy is self-financing if and only if
\begin{equation}
x_t\cdot D_t=x_0\cdot D_0+\int_0^tx_u\cdot dD_u,
\end{equation}
which is, symbolizing $d$ It\^o's ''differential'', equivalent to
\begin{equation}\label{D}
\mathfrak{D}(x_t\cdot D_t)=x_t\cdot \mathfrak{D}D_t.
\end{equation}
The selfinancing condition can be expressed by means of the anticipative ''differential'' $d_*$ as
\begin{equation}
x_t\cdot D_t=x_0\cdot D_0+\int_0^tx_u\cdot d_*D_u-\int_0^td\left<x,D\right>_u,
\end{equation}
which is equivalent to
\begin{equation}\label{DStar}
\mathfrak{D}_*(x_t\cdot D_t)=x_t\cdot \mathfrak{D}_*D_t-\mathfrak{D}_*\left<x,D\right>_t.
\end{equation}
By summing equations (\ref{D}) and (\ref{DStar}) we obtain
\begin{equation}
\mathcal{D}(x_t\cdot D_t)=\frac{1}{2}(\mathfrak{D}+\mathfrak{D}_*)(x_t\cdot D_t)=x_t\cdot \mathcal{D}D_t-\frac{1}{2}\mathfrak{D}_*\left<x,D\right>_t.
\end{equation}
To prove the second statement in expression \ref{sf} we consider the integration by part formula for It\^o's integral
\begin{equation}
\int_0^tx_u\cdot dD_u+\int_0^tD_u\cdot dx_u=x_t\cdot D_t-x_0\cdot D_0-\left<x,D\right>_t,
\end{equation}
which, expressed in terms of Stratonovich's integral, leads to
\begin{equation}
\int_0^tx_u\circ dD_u-\frac{1}{2}\left<x,D\right>_t+\int_0^tD_u\circ dx_u-\frac{1}{2}\left<x,D\right>_t=x_t\cdot D_t-x_0\cdot D_0-\left<x,D\right>_t.
\end{equation}
By taking Stratonovich's derivative $\mathcal{D}$ on both side we get
\begin{equation}
\mathcal{D}(x_t\cdot D_t)=\mathcal{D}x_t\cdot D_t+x_t\cdot \mathcal{D}D_t,
\end{equation}
which, together with the first statement in expression (\ref{sf}) proves the second one.
\end{proof}

For the reminder of this paper unless otherwise stated we will deal
only with weak $\mathcal{D}$ differentiable market models, weak $\mathcal{D}$
differentiable strategies, and, when necessary, with weak $\mathcal{D}$
differentiable state price deflators. All It\^{o} processes are weak
$\mathcal{D}$ differentiable, so that the class of considered
admissible strategies is very large.

\subsection{Arbitrage as Curvature}
 The Lie algebra of $G$ is
\begin{equation}\mathfrak{g}=\mathbf{R}^{[0, +\infty[}\end{equation}
and therefore commutative. The $\mathfrak{g}$-valued connection
$1$-form writes as
\begin{equation}\chi(x,t,g)(\delta x, \delta t)=\left(\frac{D_t^{\delta x}}{D_t^x}-r_t^x\delta t\right)g,\end{equation}
or as a linear combination of basis differential forms as
\begin{equation}\boxed{\chi(x,t,g)=\left(\frac{1}{D_t^x}\sum_{j=1}^ND_t^jdx_j-r_t^xdt\right)g.}\end{equation}
The $\mathfrak{g}$-valued curvature $2$-form is defined as
\begin{equation}R:=d\chi+[\chi,\chi],\end{equation} meaning by this,
that for all $(x,t,g)\in \mathcal{B}$ and for all $\xi,\eta\in
T_{(x,t)}M$
\begin{equation}R(x,t,g)(\xi,\eta):=d\chi(x,t,g)(\xi,\eta)+[\chi(x,t,g)(\xi),\chi(x,t,g)(\eta)]. \end{equation}
Remark that, being the Lie algebra commutative, the Lie bracket
$[\cdot,\cdot]$ vanishes. After some calculations we obtain
\begin{equation}\boxed{R(x,t,g)=\frac{g}{D_t^x}\sum_{j=1}^ND_t^j\left(r_t^x+\mathcal{D}\log(D_t^x)-r_t^j-\mathcal{D}\log(D_t^j)\right)dx_j\wedge dt,}\end{equation}
and can prove following results which characterizes arbitrage as
curvature.
\begin{theorem}[\textbf{No Arbitrage}]\label{holonomy}
The following assertions are equivalent:
\begin{itemize}
\item [(i)] The market model satisfies the no-free-lunch-with-vanishing-risk condition.
\item[(ii)] There exists a positive martingale $\beta=(\beta_t)_{t}$ such that deflators and short rates satisfy for all portfolio nominals and all times the condition
\begin{equation}\label{diffarb}\boxed{r_t^x=-\mathcal{D}\log(\beta_tD_t^x).}\end{equation}
\item[(iii)] There exists a positive martingale $\beta=(\beta_t)_{t}$ such that deflators and term structures satisfy for all portfolio nominals and all times the condition
\begin{equation}\label{intarb}\boxed{P^x_{t,s}=\frac{\mathbb{E}_t[\beta_sD^x_s]}{\beta_tD^x_t}.}\end{equation}
\end{itemize}
The following assertions are equivalent and follow from the above ones:
\begin{itemize}
\item[(iv)] The local holonomy group $\Hol^0(\chi)$ of the principal fibre bundle
$\mathcal{B}$ is trivial.
\item[(v)] The curvature form $R$ vanishes everywhere on $\mathcal{B}$.
\end{itemize}
\end{theorem}
\begin{proof}\text{}
\begin{itemize}
\item (i)$\Leftrightarrow$(iii): By Theorems
\ref{ThmDeSch} and \ref{ThmZ} the
no-free-lunch-with-vanishing-risk property is equivalent to the
existence of a positive state price deflator, that is of a
positive martingale $\beta=(\beta_t)_{t}$ such that the
market value at time $t$ of the any contingent claim  at time $s>t$ of the form $D_s^x$
is
\begin{equation}\frac{1}{\beta_t}\mathbb{E}_{t}[\beta_{s}D_s^x],\end{equation}
where $\mathbb{E}_t$ denotes conditional expectation. Since prices
are expressed in units of the deflator to which they relate the
formula writes
\begin{equation}D_t^xP_{t,s}^x=\frac{1}{\beta_t}\mathbb{E}_{t}[\beta_{s}D_s^x].\end{equation}
\item (ii)$\Leftrightarrow$(iii): the first equation is the integral version of the second, which is the differential version of the first. This can be seen by the following reasoning. Differentiating equation (\ref{intarb}) with respect to $\mathcal{D}_s$ on both sides leads to
\begin{equation}
D^x_t\exp\left(\int_t^sf^x_{t,u}du\right)(-f_{t,s}^x)=\frac{1}{\beta_t}\mathcal{D}_s\mathbb{E}_t[\beta_sD_s^x].
\end{equation}
By taking the limit for $s\rightarrow t^+$ on both sides, one obtains, by continuity,
\begin{equation}
-D_t^xr_t^x=\frac{1}{\beta_t}\mathcal{D}_t(\beta_tD_t^x),
\end{equation}
which is equation (\ref{diffarb}). This proves the implication $(ii)\Rightarrow(iii)$.\\
If we integrate (\ref{diffarb}), we obtain
\begin{equation}
-\int_t^sdu\,r_u^x=\log\left(\frac{\beta_sD_s^x}{\beta_t D_t^x}\right),
\end{equation}
which means
\begin{equation}
D^x_t\exp\left(-\int_t^sdu\,r_u^x\right)=\frac{1}{\beta_t}\beta_sD_s^x.
\end{equation}
Taking the conditional expectation $\mathbb{E}_t[\cdot]$ on both sides leads to
\begin{equation}
D^x_t\mathbb{E}_t\left[\exp\left(-\int_t^sdu\,r_u^x\right)\right]=\frac{1}{\beta_t}\mathbb{E}_t[\beta_sD_s^x],
\end{equation}
which is equivalent to equation(\ref{intarb}). Therefore, the implication $(iii)\Rightarrow(ii)$ is proved.
\item (ii)$\Rightarrow$(v):
\begin{equation}  \mathcal{D}\log(D_t^x)+r_t^x =-\mathcal{D}\log\beta_t=: \text{Const}_t\Rightarrow R\equiv0,\end{equation} where $\text{Const}_t$ depends only on the time
$t$ but not on the portfolio $x$.
\item (iv)$\Leftrightarrow$(v): The bundle is trivial. The
assertion is then a standard result in differential geometry (the
Ambrose-Singer Theorem), see f.i. \cite{KoNo96} Chapters II.4 and
II.8.
\end{itemize}
\end{proof}
\noindent The preceding Theorem motivates the following definition
\begin{defi}
The market model satisfies the \textbf{zero curvature condition (ZC)} if and only if
the curvature vanishes a.s.
\end{defi}
\noindent Therefore, we have following implications relying the three different definitions of no-arbitrage:
\begin{corollary}
\begin{equation}
\boxed{
  \begin{split}
       &\text{(NFLVR)}\Rightarrow \text{(NA)}
       \\
       &\text{(NFLVR)}\Rightarrow\text{(ZC)}
  \end{split}
        }
\end{equation}
\end{corollary}
\begin{proof}
The first implication is well known in mathematical finance (see Definition \ref{def1}). That (NFLVR) implies (ZC) is a consequence of Theorem \ref{holonomy}.
\end{proof}

\subsection{No Arbitrage Condition, Flows and Continuity Equation}
The geometric language introduced enlightens similarities between
the asset model on one side and hydro- or electrodynamics on the
other. The counterpart of a liquid or charge flow in physics is a
value flow in mathematical finance. The associated continuity
equation is satisfied if and only if the
no-free-lunch-with-vanishing-risk condition is fulfilled.
\begin{definition}
Let $M_t(x):=\{y|(y,t)\in M,\,y\le x\}$ be the set of all possible
portfolios at time $t$ bounded from above by the portfolio $y$, and $M_t^j(x)\subset ]-\infty,x_j]$ its projection onto the $j$th axis.
The \textbf{log value current}  for the market model is defined as
a vector field $J$ on $M$ by
\begin{equation}\boxed{J^j(x,t):=\left(\int_{M_t^j(x)} dy_j\,\frac{y_jD_t^j}{D_t^{(y_je_j+\sum_{i\neq j}x_ie_i)}}\right)r_t^j 
,}
\end{equation} where
$r_t:=[r_t^1,\dots,r_t^N]^{\dagger}$ and $r_t\,u $ the componentwise
multiplication. Let $\beta=(\beta_t)_t$ be a positive
semimartingale. The $\beta$-scaled \textbf{log value density} for
the market model is defined on $M$ as
\begin{equation}\boxed{\rho^{\beta}(x,t):=\log(\beta_tD_t^x).}\end{equation}
\end{definition}
The results of the preceding subsection can be reformulated in terms
of a \textit{continuity} equation analogously to classical
electrodynamics (cf. \cite{Ja98}, 5.1, p. 175) or continuum
mechanics (cf. \cite{Si02}, 3.3.1, pp. 67-68).
\begin{theorem}[\textbf{Continuity Equation}]
The market model satisfies the
no-free-lunch-with-vanishing-risk-condition if and only if
 there exists a positive martingale $\beta$ such that one of
following equations is satisfied:
\begin{equation}\label{cont}
\boxed{
  \begin{split}
   &\mathcal{D}\rho^{\beta}+\dive_xJ=0\\
   &\mathcal{D}\int_{X_{t_0}}dx^N\rho^{\beta}=-\oint_{\partial X_{t_0}} dn\cdot J.
\end{split}
}
\end{equation}
The first expression  is the differential version of the continuity
equation and the second the integral one, which must hold text for
any $1-$codimensional submanifold $X_{t_0}$ lying in the hyperplane
$t\equiv t_0$.
\end{theorem}
The integral version of the continuity equation has a beautiful
financial interpretation: a market satisfies the no-free-lunch-with-
vanishing-risk condition if and only if the log total value change
of any submarket is due to the log value current flow through its
boundary.
\begin{proof}
It is an application of the vector field divergence definition:
\begin{equation}
\dive_xJ(x,t)=\sum_{j=1}^N\frac{\partial}{\partial
x_j}J(x,t)=r_t^x=-\mathcal{D}\log(\beta_tD_t^x).
\end{equation}
Gauss' Theorem proves the integral version.
\end{proof}
The left hand side expression in the differential version of the
continuity equation (\ref{cont}) is a natural candidate for a local
arbitrage measure. The link with the definition of arbitrage given
in the preceding section is given by
\begin{proposition}[\textbf{Curvature Formula}]\label{curvature}
Let $R$ be the curvature, $\rho^{\beta}$ the log value density and
$J$ the log value current. Then, the following quality holds:
\begin{equation}\boxed{R(x,t,g)=g dt\wedge d_x\left[\mathcal{D} \rho^{\beta}+\dive_x J\right]=g dt\wedge d_x\left[\mathcal{D} \log (D_t^x)+r_t^x\right].}\end{equation}
\end{proposition}
\begin{proof}
We develop the expression for the curvature as:
\begin{equation}
\begin{split}
R(x,t,g) &= \frac{g}{D_t^x}\sum_{j=1}^ND_t^j\left(r_t^x+\mathcal{D}\log(D_t^x)-r_t^j-\mathcal{D}\log(D_t^j)\right)dx_j\wedge dt=\\
&=g\sum_{j=1}^N\frac{\partial}{\partial x_j}\left(-\mathcal{D}\log(D_t^x)-r_t^x\right)dx_j\wedge dt=\\
&=g\sum_{j=1}^N\frac{\partial}{\partial x_j}\left(-\mathcal{D}\log(\beta_tD_t^x)-r_t^x\right)dx_j\wedge dt=\\
&=g\sum_{j=1}^N\frac{\partial}{\partial
x_j}\left(\mathcal{D} \rho^{\beta} +  \dive_x (J)\right)dt\wedge dx_j=\\
&=g dt\wedge d_x\left[\mathcal{D}\rho^{\beta}+\dive_x J\right]=\\
&=gdt\wedge d_x\left[\mathcal{D}\log(D_t^x)+r_t^x\right].
\end{split}
\end{equation}
\end{proof}
\begin{corollary}[\textbf{No Arbitrage Revisited}]
The following assertions are equivalent:
\begin{itemize}
\item [(i)] The market model satisfies the no-free-lunch-with-vanishing-risk condition.
\item[(ii)] There exist a positive martingale $\beta$ for which the continuity equation is satisfied.
\end{itemize}
\end{corollary}

\section{A Guiding Example}\label{GuidEx}
We want now to construct an example to demonstrate how the most
important geometric concepts of section \ref{section2} can be
applied. Given a filtered probability space $(\Omega,\mathcal{A},
P)$, where $P$ is the statistical (physical) probability measure,
we assume that all processes introduced in this example are
adapted to the filtration
$\mathcal{A}=(\mathcal{A}_t)_{t\in[0,+\infty[}$ satisfying the
usual conditions. Let us consider a market consisting of  $N+1$
assets labeled by $j=0,1,\dots,N$, where the $0$-th asset is the
cash account utilized as a num\'{e}raire. Therefore, as explained
in the introductory subsection \ref{StochasticPrelude}, it
suffices to model the price dynamics of the other assets
$j=1,\dots,N$ expressed in terms of the $0$-th asset. As vector
valued semimartingale for the discounted price process
$\hat{S}:[0,+\infty[\times\Omega\rightarrow\mathbf{R}^N$, we chose
the
 multidimensional It\^{o}-process given by
\begin{equation}\label{SDES}
d\hat{S}_t=\hat{S}_t(\alpha_tdt+\sigma_tdW_t),
\end{equation}
where
\begin{itemize}
\item $(W_t)_{t\in[0,+\infty[}$ is a standard $P$-Brownian motion in $\mathbb{R}^K$, for some $K\in\mathbb{N}$, and,
\item $(\sigma_t)_{t\in[0,+\infty[}$,  $(\alpha_t)_{t\in[0,+\infty[}$ are  $\mathbb{R}^{N\times K}$-, and respectively,  $\mathbb{R}^{N}$- valued  stochastic processes, $\sigma_t$ as maximal rank, i.e. $\text{rank}(\sigma_t)=K$,
\end{itemize}
The processes $\alpha$ and $\sigma$ generalize  drift and volatility
of a multidimensional geometric Brownian motion. Therefore, we have
modelled assets satisfying the zero liability assumptions like
stocks, bonds and commodities. The solution of the SDE (\ref{SDES})
can be obtained by means of It\^{o}'s Lemma and reads
\begin{equation}
\hat{S}_t=\hat{S}_0\exp\left(\int_0^t\left(\alpha_u-\frac{1}{2}\diag(\sigma_u\sigma_u^{\dagger})\right)du+\int_0^t\sigma_udW_u
\right),
\end{equation}
\noindent where integration and exponentiation are meant
componentwise. To define the corresponding deflators to meet
Definition \ref{defi1}, we can just set
\begin{equation}
D:=\hat{S}.
\end{equation}
In order to construct term structures representing future contracts
on the assets, we pass by the definition of their short rates as in
Definition \ref{int}, assuming that they follow the multidimensional
It\^{o}-process
\begin{equation}\label{SDEr}
dr_t=a_tdt+b_tdW_t,
\end{equation}
where $W$ is the multidimensional $P$-Brownian motion introduced above
and $(b_t)_{t\in[0,+\infty[}$,  $(a_t)_{t\in[0,+\infty[}$ are
$\mathbf{R}^{N\times K}$-, and respectively,  $\mathbf{R}^{N}$-
valued locally bounded predictable stochastic processes, the drift
and the instantaneous volatility of the multidimensional short rate.
The solution of the SDE (\ref{SDEr}) writes
\begin{equation}
r_t=r_0+\int_0^ta_udu+\int_0^tb_udW_u.
\end{equation}
Term structures are defined via
\begin{equation}P_{t,s}:=\mathbb{E}_t\left[\exp\left(-\int_t^sr_udu\right)\right].\end{equation}
 At time $t$, the price of synthetic zero bonds delivering at time $s$
one unit of the base asset $j$ is
\begin{equation}
\bar{S}_t^j:=\hat{S}_tP^j_{t,s}.
\end{equation}
This means that we have constructed $N$ gauges
\begin{equation}
(D^j,P^j)\quad j=1,\dots,N,
\end{equation}
satisfying Definition \ref{defi1}. Moreover, if drifts $\alpha, a$
and volatilities $\sigma, b$ satisfy appropriate regularity
assumptions, then we have a Nelson $\mathcal{D}$ differentiable market model as in
Definition \ref{weakMM} with  nominal space $X=\mathbf{R}^N$ and
base manifold $M=\mathbf{R}^N\times[0,+\infty[$. The dynamics of
asset prices, short rates and term structures read
\begin{equation}\label{dyn}
\begin{split}
\hat{S}_t^x&=x^{\dagger}\hat{S}_0\exp\left(\int_0^t(\alpha_u-\frac{1}{2}\diag({\sigma_u}\sigma_u^{\dagger}))du+\int_0^t\sigma_udW_u\right)\\
r_t^x&=\frac{x^{\dagger}}{\hat{S}_t^x}\hat{S}_t(r_0+\int_0^ta_udu+\int_0^tb_udW_u)\\
P_{t,s}^x&=\mathbb{E}_t\left[\exp\left(-\int_t^sr_u^xdu\right)\right],
\end{split}
\end{equation}
for any nominals $x\in\mathbf{R}^N$. The curvature of the
market principal fibre bundle $\mathcal{B}$ can be computed with
Proposition \ref{curvature}:
\begin{equation}\label{curvEx}
R(x,t,g)=gdt\wedge d_x\left(\mathcal{D}\log \hat{S}_t^x+r_t^x\right)
\end{equation}
The zero curvature condition is equivalent to
\begin{equation}\label{exprr}
\mathcal{D}\log \hat{S}_t^x+r_t^x=C_t,
\end{equation}
where $C$ is a stochastic processes which does not depend on $x$.
Inserting equation (\ref{exprr}) into the expression for the short
rate in equation (\ref{dyn}) allows us to compute the term structure
as
\begin{equation}\label{raw}
P_{t,s}^x=\mathbb{E}_t\left[\exp\left(-\int_t^sr_u^xdu\right)\right]=\mathbb{E}_t\left[\frac{\hat{S}_s^x}{\hat{S}^x_t}\frac{\beta_s}{\beta_t}\right],
\end{equation}
\noindent where we have introduced  the positive stochastic process
\begin{equation}
\beta_t:=\exp\left(-\int_0^tC_udu\right).
\end{equation}
Equation (\ref{raw}) can be rewritten as
\begin{equation}\label{mar}
\beta_t\hat{S}^x_tP_{t,s}^x=\mathbb{E}_t\left[\beta_s\hat{S}_s^xP_{s,s}^x\right],
\end{equation}
meaning that for the price of the synthetic zero bond
\begin{equation}
\bar{S}_t^{x,T}:=\hat{S}_t^xP^x_{t,T},
\end{equation}
the process $(\beta_t\bar{S}_t^{x,T})_{t\in[0,T[}$ is a
$P$-martingale for all maturities $T\in[0,+\infty[$. Therefore, if
the positive stochastic process $\beta$ is a martingale, then
it is a pricing kernel and the no-free-lunch-with-vanishing-risk
condition is satisfied. here below we will investigate under what
conditions this is the case. Conversely, from (NFLVR) one can
infer the vanishing of the curvature. We have thus rediscovered
Theorem \ref{holonomy}.

\begin{proposition}\label{PropIto}
Let the dynamics of a market model be specified by following It\^o's processes as in (\ref{SDES}, \ref{SDEr}), where we additionally assume that  the coefficients
\begin{itemize}
\item $(\alpha_t)_t,(\sigma_t)_t$, and  $(r_t)_t$ satisfy
\begin{equation}
\lim_{s\rightarrow t^+}\mathbb{E}_s[\alpha_t]=\alpha_t,\quad\lim_{s\rightarrow t^+}\mathbb{E}_s[r_t]=r_t,\quad\lim_{s\rightarrow t^+}\mathbb{E}_s[\sigma_t]=\sigma_t,
\end{equation}
\item $(\sigma_t)_t$ is an It\^o's process,
\item $(\sigma_t)_t$ and $(W_t)_t$ are independent processes.
\end{itemize}
Then, the market model satisfies the (ZC) condition if and only if
\begin{equation}\label{ZCCond}
\alpha_t+r_t \in {\rm Range} (\sigma_t).
\end{equation}
\end{proposition}
\begin{remark}In the case of the classical model, where there are no term structures (i.e. $r\equiv0$), the condition (\ref{ZCCond}) reads as $\alpha_t\in{\rm Range} (\sigma_t)$.
\end{remark}

\begin{proof}
Let us consider the expression for It\^o's integral with respect to Stratonovich's
\begin{equation}
\int_0^t\sigma_udW_u=\int_0^t\sigma_u\circ dW_u-\frac{1}{2}\int_0^td\left<\sigma,W\right>_u,
\end{equation}
and take Nelson's derivative corresponding to the Stratonovich's integral:
\begin{equation}
\mathcal{D}\int_0^t\sigma_udW_u=\sigma_t\mathcal{D}W_t-\frac{1}{2}\mathcal{D}\left<\sigma,W\right>_t.
\end{equation}
Since
\begin{equation}
\mathcal{D}W_t=\frac{W_t}{2t}
\end{equation}
and, because of the independence assumption for the two It\^{o}s processes $(\sigma_t)_t$ and $(W_t)_t$,
\begin{equation}
\left<\sigma,W\right>_t\equiv0,
\end{equation}
we obtain
\begin{equation}
\mathcal{D}\int_0^t\sigma_udW_u=\sigma_t\frac{W_t}{2t},
\end{equation}
which, inserted into the asset dynamics
\begin{equation}
\hat{S}_t=\hat{S}_0\exp\left(\int_0^t(\alpha_u-\frac{1}{2} {\rm diag} ({\sigma_u}\sigma_u^{\dagger}))du+\int_0^t\sigma_udW_u\right),
\end{equation}
\noindent leads to
\begin{equation}
\mathcal{D}\log\hat{S}_t=\alpha_t-\frac{1}{2} {\rm diag}({\sigma_t}\sigma_t^{\dagger})+\sigma_t\frac{W_t}{2t}.
\end{equation}
By Proposition \ref{curvature} the curvature vanishes if and only if for all $x\in\mathbb{R}^N$
\begin{equation}
\mathcal{D}\log \hat{S}_t^x+r_t^x=C_t,
\end{equation}
for a real valued stochastic process $(C_t)_{t\ge0}$, or, equivalently
\begin{equation}
\mathcal{D}\log \hat{S}_t+r_t=C_te,
\end{equation}
where $e:=[1,\dots,1]^{\dagger}$ or
\begin{equation}\label{noarbalphabeta}
\alpha_t+r_t-\frac{1}{2} {\rm diag} (\sigma_t{\sigma_t}^{\dagger})+\sigma_t\frac{W_t}{2t}=C_te.
\end{equation}
Equation (\ref{noarbalphabeta}) is the formulation of the (ZC) condition for the market model (\ref{SDES}).
By taking on both sides of (\ref{noarbalphabeta}) $\lim_{h\rightarrow0^+}\mathbb{E}_{t-h}[\cdot]$, and utilizing the independence assumption, from which
\begin{equation}
\mathbb{E}_{t-h}\left[\sigma_t\frac{W_t}{2t}\right]=\mathbb{E}_{t-h}\left[\sigma_t\right]\underbrace{\mathbb{E}_{t-h}\left[\frac{W_t}{2t}\right]}_{=0}=0
\end{equation}
follows, we obtain, using the continuity assumption for $(\alpha_t)_t,(\sigma_t)_t$, and  $(r_t)_t$,
\begin{equation}
\alpha_t+r_t-\frac{1}{2} {\rm diag}(\sigma_t{\sigma_t}^{\dagger})=\beta_te,
\end{equation}
where $\beta_t:=\lim_{h\rightarrow0^+}\mathbb{E}_{t-h}[C_t]$ is a predictable process. Therefore, equation (\ref{noarbalphabeta}) becomes
\begin{equation}
\sigma_t\frac{W_t}{2t}=(C_t-\beta_t)e,
\end{equation}
and, thus
\begin{equation}\
e\in{\rm Range} (\sigma_t),
\end{equation}
the space spanned by the column vectors of $\sigma_t$. Since $\sigma_t$ has maximal rank, the $K$column vectors of $\sigma_t$ are linearly independent and $C_t-\beta_t\neq0$.\par Let $P_{\sigma_t}$, $P_{\sigma_t^\bot}$ denote the orthogonal projections onto ${\rm Range} (\sigma_t)$ and its orthogonal complement ${\rm Range} (\sigma_t)^{\bot}$, respectively. Then, we can decompose
\begin{equation}
\alpha_t+r_t=P_{\sigma_t}(\alpha_t+r_t)+P_{\sigma_t^\bot}(\alpha_t+r_t),
\end{equation}
and
\begin{equation}\label{eqP}
P_{\sigma_t^\bot}(\alpha_t+r_t)=P_{\sigma_t^\bot}\left(C_te\right)-P_{\sigma_t^\bot}\left(\sigma_t\frac{W_t}{2t}\right)+P_{\sigma_t^\bot}\left(\frac{1}{2} {\rm diag} (\sigma_t{\sigma_t}^{\dagger})\right).
\end{equation}
Since $e$ and $\sigma_tW_t$ lie in ${\rm Range} (\sigma_t)$, the first two addenda on the right hand side of (\ref{eqP}) vanish. By Lemmata \ref{diag} and \ref{proj} the third one vanishes as well, so that $P_{\sigma_t^\bot}(\alpha_t+r_t)=0$, i.e. $\alpha_t+r_t\in{\rm Range} (\sigma_t)$. Conversely, if  $\alpha_t+r_t\in{\rm Range} (\sigma_t)$, then equation (\ref{noarbalphabeta}) holds true, and
the proof of the equivalence between the (ZC) condition and (\ref{ZCCond}) is completed.
\end{proof}

\begin{lemma}\label{diag}
Let $A$ be a linear operator on the euclidean $\mathbf{R}^N$. The vector
\begin{equation}
\diag(A):=\sum_{j=1}^N(Ae_j\cdot e_j)e_j
\end{equation}
does not depend on the choice of the o.n.B $\{e_1,\dots,e_n\}$ of $\mathbf{R}^N$ and defines the \textbf{diagonal} of $A$.
\end{lemma}
\begin{proof}
The coordinates of $\diag(A)$ with respect to the o.n.B $\{e_1,\dots,e_N\}$ can be written as
\begin{equation}
[\diag(A)]_{\{e\}}=\sum_{j=1}^N([e_j]_{\{e\}}^\dagger[A]_{\{e\}}[e_j]_{\{e\}})[e_j]_{\{e\}}
\end{equation}
Let us consider another o.n.B $\{f_1,\dots,f_n\}$ of $\mathbf{R}^N$. This means that there exists an orthogonal linear operator $U$ on $\mathbf{R}^N$ such that $Ue_j=f_j$ for all $j=1,\dots,N$. Therefore we can write
\begin{equation}
\begin{split}
[\diag(A)]_{\{e\}}&=\sum_{j=1}^N\left(([U]_{\{f\}}^\dagger[f_j]_{\{f\}})^\dagger[A]_{\{e\}}[U]_{\{f\}}^\dagger[f_j]_{\{f\}}\right)[U]_{\{f\}}^\dagger[f_j]_{\{f\}}=\\
&=\sum_{j=1}^N\left([f_j]_{\{f\}}^\dagger\left([U]_{\{f\}}[A]_{\{e\}}[U]_{\{f\}}^\dagger\right)[f_j]_{\{f\}}\right)[U]_{\{f\}}^\dagger[f_j]_{\{f\}}=\\
&=[U]_{\{f\}}^\dagger\left(\sum_{j=1}^N[f_j]_{\{f\}}^\dagger[A]_{\{f\}}[f_j]_{\{f\}}\right)=\\
&=[U]_{\{f\}}^\dagger[\diag(A)]_{\{f\}}.
\end{split}
\end{equation}
Therefore, the coordinates of the diagonal transforms like a vector during a change of basis, and hence the diagonal is well defined.\\
\end{proof}
\begin{lemma}\label{proj}
Let $\sigma$ be a $\mathbf{R}^{N\times K}$ real matrix of rank $K$ and $P$ the orthogonal projection onto the orthogonal complement to the subspace generated by the column vectors of $\sigma$. Then,
\begin{equation}
P\diag(\sigma\sigma^{\dagger})=0\in\mathbf{R}^N.
\end{equation}
\end{lemma}
\begin{proof}
The real symmetric matrix $\sigma\sigma^{\dagger}\in\mathbf{R}^{N\times N}$ induces via standard o.n.b a selfadjoint linear operator on $\mathbf{R}^N$, which by Lemma \ref{diag} has a well defined diagonal. Let us enlarge $\sigma$ to an $\mathbf{R}^{N\times N}$ matrix, by adding $N-K$ zero column vectors. The matrix $\sigma\sigma^{\dagger}\in\mathbf{R}^{N\times N}$ remains the same. Let us consider an o.n.b of $\mathbf{R}^N$, $\{f_1,\dots,f_N\}$, where $\{f_1,\dots,f_K\}$ is a basis of $\text{Range}(\sigma)$ and $\{f_{K+1},\dots,f_N\}$ is a basis of its orthogonal complement, $\text{Range}(\sigma)^\bot$. The diagonal of $\sigma\sigma^{\dagger}$ reads
\begin{equation}
\diag(\sigma\sigma^{\dagger})=\sum_{j=1}^N(\sigma\sigma^{\dagger}f_j\cdot f_j)f_j=\sum_{j=1}^N(\sigma^{\dagger}f_j\cdot \sigma^{\dagger}f_j)f_j=\sum_{j=1}^K(\sigma^{\dagger}f_j\cdot \sigma^{\dagger}f_j)f_j,
\end{equation}
because $\sigma^{\dagger}f_j=0$ for $j=K+1,\dots,N$, being $f_j$ in the orthogonal complement of $\text{Range}(\sigma)$. Therefore,
\begin{equation}
P\diag(\sigma\sigma^{\dagger})=\sum_{j=1}^K(\sigma^{\dagger}f_j\cdot \sigma^{\dagger}f_j)Pf_j=0,
\end{equation}
because $f_j$ is in $\text{Range}(\sigma)$ for $j=1,\dots,K$ and $P$ is the projection onto $\text{Range}(\sigma)^\bot$ .
\end{proof}
\noindent Next, we show the equivalence of the (ZC) condition with (NFLVR) in the case of It\^o's dynamics.
\begin{proposition}\label{NovikovThm}
Under the same assumptions as Proposition \ref{PropIto}, the zero curvature condition for the market model specified by (\ref{SDES}, \ref{SDEr}) , that is
\begin{equation}
\mathcal{D}\log \hat{S}_t+r_t =C_te,
\end{equation}
is equivalent to the no-free-lunch-with-vanishing-risk condition if the positive stochastic process $(\beta_t)_{t\ge0}$, defined as
\begin{equation}
\beta_t:=\exp\left(-\int_0^tC_udu\right)
\end{equation}
is a martingale.
\end{proposition}
\begin{proof}
By Proposition \ref{curvature} the zero curvature (ZC) condition $R=0$ is equivalent with the existence of a stochastic process $(C_t)_{t\ge0}$ such that for all $i=1,\dots,N$ the equation
\begin{equation}
\mathcal{D}\log \hat{S}_t^i+r_t ^i=C_t
\end{equation}
\noindent holds. This means that
\begin{equation}
\begin{split}
&\mathcal{D}\log \hat{S}_t^i=C_t-r_t^i\\
&\log\frac{S_t^i}{S_0^i}=\int_0^t(C_u-r_u^i)du\\
&S_t^i=S_0^i\exp\left(\int_0^tC_udu\right)\exp\left(-\int_0^tr_u^idu\right).
\end{split}
\end{equation}
Therefore,
\begin{equation}
\mathcal{D}\log(\beta_tD_t^i)+r_t^i=0
\end{equation}
for all $i=1,\dots,N$. By Theorem \ref{holonomy}, if $(\beta_t)_{t\ge0}$ is a martingale, then we have proved (NFLVR).
\end{proof}

\begin{proposition}
For the market model whose dynamics is specified by the SDEs
\begin{equation}
\begin{split}
d\hat{S}_t&=\hat{S}_t(\alpha_tdt+\sigma_tdW_t)\\
dr_t&=a_tdt+b_tdW_t,
\end{split}
\end{equation}
the no-free-lunch-with-vanishing risk condition (NFLVR) is satisfied if \textbf{Novikov's condition}
\begin{equation}\label{Novikov}
\mathbb{E}_0\left[\exp\left(\int_0^T\frac{1}{2}\left|{\sigma_t}^{\dagger}(\sigma_t{\sigma_t}^{\dagger})^{-1}(\alpha_t+r_t)\right|^2du\right)\right]<+\infty,
\end{equation}
is fulfilled.
\end{proposition}
\begin{proof}
The asset price dynamics reads
\begin{equation}\label{eqS}
d\log \hat{S}_t =\alpha_t dt+\sigma_tdW_t.
\end{equation}
Since
\begin{equation}
P_{t,s}=\exp\left(-\int_t^sf_{t,u}du\right),
\end{equation}
the term structure dynamics reads
\begin{equation}\label{eqr}
d_t\log P_{t,s}=f_{t,t}dt=r_tdt,
\end{equation}
where we consider $s$ as a parameter. Putting (\ref{eqS}) and (\ref{eqr}) together leads to
\begin{equation}\label{eqS}
d\log\left(\hat{S}_tP_{t,s}\right) =(\alpha_t+r_t) dt+\sigma_tdW_t=\sigma_t \left(\underbrace{{\sigma_t}^{\dagger}(\sigma_t{\sigma_t}^{\dagger})^{-1}(\alpha_t+r_T)}_{:=\gamma_t}dt+dW_t\right)
\end{equation}
If the Novikov condition (\ref{Novikov}) is satisfied then, by
Girsanov's Theorem, the process
\begin{equation}
m_t^*:=\exp\left(-\int_0^t\frac{1}{2}\left|\gamma_u\right|^2du+
\int_0^t\gamma_u^{\dagger}dW_u\right)
\end{equation}
is a martingale and the Radon-Nykodym derivative of a probability measure $\mathbb{P}^*$ equivalent to the statistical probability measure $\mathbb{P}$:
\begin{equation}
\mathbb{E}_t\left[\frac{d\mathbb{P}^*}{d\mathbb{P}}\right]=m_t^*.
\end{equation}
Therefore
\begin{equation}
W_t=W^*_t+\int_0^t\gamma_udu,
\end{equation}
where $(W_t^*)_{t\ge0}$ is a $\mathbb{P}^*$ standard multivariate
Brownian motion, and
\begin{equation}
dW_t=dW^*_t+\gamma_tdt,
\end{equation}
which leads to
\begin{equation}
d\log (\hat{S}_tP_{t,s})=\sigma_tdW_t^*.
\end{equation}
We conclude that $(\hat{S}_tP_{t,s})_{t\ge0}$ is an exponential
$\mathbb{P}^*$-martingale and hence that the market satisfies
\begin{equation}
\hat{S}_tP_{t,s}=\mathbb{E}_t^*[\hat{S}_s\underbrace{P_{s,s}}_{=1}]=\mathbb{E}_t\left[\frac{m_s^*}{m_t^*}\hat{S}_s\right],
\end{equation}
which, by Theorem \ref{holonomy} is equivalent, being $(m_t^*)_t$ a positive martingale, to the no free-lunch-with-vanishing-risk.
\end{proof}
\section{Market Model and Differential Topology}

Till now we have not considered any concrete asset dynamics for the
market model. In our differential geometric framework a market
dynamics should be specified as an infinite dimensional stochastic
process for deflators $D_t:=[D^1_t,\dots,D^N_t]^{\dagger}$ and term
structures $P_{ts}:=[P_{ts}^1,\dots,P_{ts}^N]^{\dagger}$ or
equivalently as $2N-$dimensional stochastic process for deflators
and short rates $r_t:=[r^1_t,\dots,r^N_t]^{\dagger}$. Inspired by
\cite{Il01}, we modify Ilinski's five principles characterizing the
market dynamics to obtain
\begin{definition}[\textbf{Principles of Market
Dynamics}]\label{axioms}\text{}
\begin{itemize}
\item \textbf{(A1) Intrinsic Uncertainty: } Deflators, term structures and short rates are random
variables:
\begin{equation}D=D(\omega),\quad P=P(\omega),\quad r=r(\omega), \end{equation}
where $\omega\in \Omega$ represents a state of nature and
$(\Omega,\mathcal{A},P)$ is a probability space.

\item \textbf{(A2) Causality :} Time dynamics of deflators,
term structures and short rates depend on their history such that
future events can be influenced by past events only. Formally, we
assume the existence of a filtration $(\mathcal{A}_t)_{t\ge 0}$ of
the $\sigma$-algebra $\mathcal{A}_{\infty}$ such that $D_t$, $r_t$
and $(P_{t,s})_{s\ge t}$ are $\mathcal{A}_t$-measurable.

\item \textbf{(A3) Gauge Invariance: }Assume that deflator-term-structure stochastic process for $(D,P)=((D_t,P_{t,s}))_{t\ge
0, s\ge t}$ satisfies the equation
\begin{equation}f(D,P)=0\end{equation} almost surely. Thereby, $f$ denotes a possibly stochastic function, that is  $f=f(\omega,D,P)$.
Then, for every invertible gauge transform $\pi$ the equation
\begin{equation}f(D^{\pi},P^{\pi})=0\end{equation} must hold almost surely as well.

\item \textbf{(A4) Minimal Arbitrage: } The most likely configurations
of the random connections among deflators and term structures (or
short rates) are those minimizing the arbitrage for the market
portfolio strategy for almost every state of the nature
$\omega\in\Omega$.

\item \textbf{(A5) Extension Consistency :} The theory has to
contain stochastic finance theory.
\end{itemize}
\end{definition}

\noindent We will try to realize this program in the rest of this
paper. We remark that our framework already fulfills principles
$(A1)$ and $(A2)$. In section \ref{LagrangianTheory} we will obtain
$(A3)$, $(A4)$ and $(A5)$ .

\subsection{Arbitrage Action as Homotopic Invariant}
We will investigate what happens to arbitrage properties of
self-financing strategies in case of smooth deformations. We will
assume for the remaining of this paper that the space of allocations
$X$ is connected.

\begin{definition}
Two $\mathcal{D}$-admissible self-financing strategies
$\gamma_{1,2}:[0,T]\rightarrow X$ with the same start and end points
(i.e. $\gamma_1(0)=\gamma_2(0)$ and $\gamma_1(T)=\gamma_2(T)$)are
said to be \textbf{homotopic} if and only if one is a smooth
deformation of the other. This means that there exists a
$\mathcal{D}$-differentiable function $\Gamma:[0,1]\times[0,T]\rightarrow X$
such that $\gamma_1=\Gamma(0,\cdot)$ and $\gamma_2=\Gamma(1,\cdot)$.
A $\mathcal{D}$-admissible self-financing strategy is said to be
\textbf{contractible} if it is null-homotopic that is homotopic to a
point. The relation \textbf{homotopy} is an equivalence relation in
the set of self-financing strategies and its quotient space, that is
the set of all equivalence classes, is called the first
\textbf{self-financing differentiable fundamental group} of the
portfolio space and denoted by $\Pi_1(X, D)$. The market is said to
be \textbf{simply connected} in the case of a trivial first
fundamental group or equivalently if and only if every
closed self-financing strategy is contractible.
\end{definition}
\begin{figure}[h!]
\centering
\includegraphics[width = 12.5cm, angle=-90]{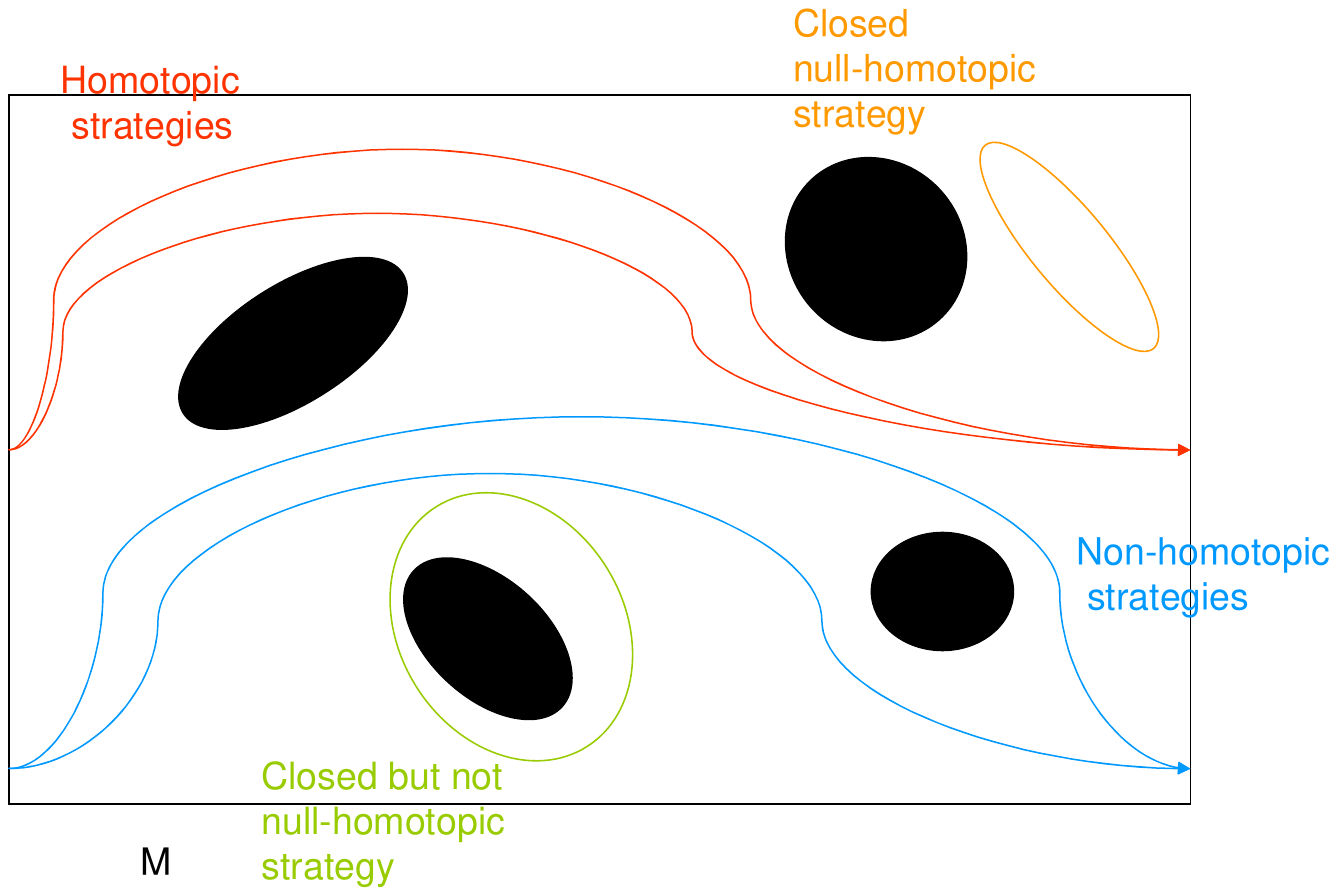}
\caption{Homotopy}\label{DH}
\end{figure}
The intuition behind this definition is that homotopic equivalent strategies generate the same quantity of arbitrage.
To see it, we introduce the following

\begin{definition}\label{defAA}
Let $\beta$ be a positive martingale. The \textbf{arbitrage
action} of the $\mathcal{D}$-differentiable strategy $\gamma$ is
defined as
\begin{equation}A^{\beta}(\gamma;D,r):=\int_{\gamma}dt\BR{\mathcal{D}\rho^{\beta}+\dive_xJ}=\int_0^1dt\BR{\mathcal{D}
\log(\beta_tD_t^{x_t})+r_t^{x_t}}.
\end{equation}
\end{definition}

\begin{theorem}[\textbf{Arbitrage Action Formula for Self-financing Strategies}]
Let $\gamma:[0,T]\rightarrow X$ be a $\mathcal{D}$-admissible
self-financing strategy and
$d_{0T}^{\gamma}:=\exp(-\int_0^T\,r_u^{x_u}du)$ the stochastic
discount factor. Then
\begin{equation}A^{\beta}(\gamma;D,r)=\log\left(\frac{\beta_TD_T^{x_T}}{\beta_0D_0^{x_0}d_{0T}^{\gamma}}\right)\end{equation}
almost surely.
\end{theorem}

\begin{proof}
We rewrite the arbitrage action as
\begin{equation}
\begin{split}
A^{\beta}(\gamma;D,r) &=\int_{\gamma}dt\BR{\mathcal{D}
\log(\beta_tD_t^{x_t})+r_t^{x_t}} =\\
&=\int_{\gamma}dt\,\mathcal{D}
\log(\beta_tD_t^{x_t})+\int_{\gamma}ds\,r_t^{x_t}=\\
&=\log\left(\frac{D_T^{x_T}\beta_T}{D_0^{x_0}\beta_0}\right)-\log(d_{0T}^{\gamma})=\log\left(\frac{\beta_TD_T^{x_T}}{\beta_0D_0^{x_0}d_{0T}^{\gamma}}\right)
\end{split}
\end{equation}
and the proof is completed.
\end{proof}

\begin{lemma}[\textbf{Homotopy Invariance Property of Stochastic Discount Factor for Self-financing Strategies}]
The stochastic discount factor  satisfies $d_{0T}^{\gamma}\equiv1$
for any contractible self-financing strategy $\gamma$. The
stochastic discount factor is a homotopy invariant, that is for any
homotopy class of $\mathcal{D}$-admissible self-financing
strategies $[\gamma]\in\Pi_1(X,D)$ the following statement holds
a.s.
\begin{equation}\gamma_{1,2}\in[\gamma]\Rightarrow d_{0T}^{\gamma_1}=d_{0T}^{\gamma_2}=d_{0T}^{\gamma}.\end{equation}
\end{lemma}

\begin{proof}
Let $\gamma$ be a closed strategy. Then, for the stochastic discount
factor we have:
\begin{equation}-\log(d_{0T}^{\gamma})=\int_0^Tdu\,r_u^{x_u}=\int_{\gamma}dt\,r_t^x=\int_{\gamma}dt\,\dive_x J=\int_{\partial \gamma}\,dn\cdot J=0,\end{equation}
since $\partial \gamma=\emptyset$. The Lemma follows.
\end{proof}

\begin{corollary}[\textbf{Homotopy Invariance Property of  Action for Self-financing Strategies}]\label{corHom}
The arbitrage action $A^{\beta}(\gamma)$ vanishes for any
contractible self-financing strategy $\gamma$. The arbitrage action
is a homotopy invariant, that is for any homotopy class of
self-financing strategies $[\gamma]\in\Pi_1(X,D)$ the following
statement holds a.s.
\begin{equation}\gamma_{1,2}\in[\gamma]\Rightarrow A^{\beta}(\gamma_1;D,r)=A^{\beta}(\gamma_2;D,r)=A^{\beta}(\gamma;D,r).\end{equation}

\end{corollary}

We can now give an alternative definition of arbitrage strategy end
extend it to positive and negative arbitrage.

\begin{definition}
A $\mathcal{D}$-admissible self-financing strategy $\gamma$ is
said to be a \textbf{positive, zero, negative $\beta$-arbitrage
strategy} if and only if the arbitrage action $A^{\beta}(\gamma)$ is
a.s. positive, zero, negative, respectively.
\end{definition}

A zero arbitrage strategy is of course a no-arbitrage strategy in
the usual sense. A strategy $x$ is a negative $\beta$-arbitrage
strategy if and only if $-x$ is a positive $\beta$-arbitrage
strategy.

\begin{corollary}[\textbf{Arbitrage Strategies}]
A $\mathcal{D}$-admissible self-financing strategy $\gamma:[0,T]\rightarrow X$ is a
positive, zero, negative $\beta$-arbitrage strategy if and only if
\begin{equation}\beta_TD_T^{x_T}(\omega)\left\{
                     \begin{array}{ll}
                       > \\
                       = \\
                       <
                     \end{array}
                   \right\}
\beta_0D_0^{x_0}(\omega)d_{0T}^{\gamma}(\omega)\quad \text{
a.s.}\end{equation}
\end{corollary}

\subsection{Parametrization of Strategies and Differential-Topological Version of the Fundamental Arbitrage Pricing Theorem}

Now we will investigate the relationships between homotopy for
self-financing strategies and holonomy of the connection for the
principal fibre bundle describing the market model. Let us assume
that we have fixed a num\'{e}raire by choosing an appropriate global
cross section $g^{\text{Num}}$ of the gauge bundle $\mathcal{B}$.
From the Ambrose-Singer Theorem (see \cite{St82}, Theorem VII.1.2.)
we now that the Lie algebra $\LAhol(\chi)$ of the  holonomy group
$\Hol(\chi)$ is spanned by the values of the curvature form. More
exactly, assuming that both $M$ and $\mathcal{B}$ are connected, for
any $b\in\mathcal{B}$
\begin{equation}
\begin{split}
\LAhol(\chi)&=\spa\left(\left\{R(c)(\eta,\xi)\mid\eta,\xi\in T_{c}\mathcal{B},\text{ for an horizontal curve}\right.\right.\\
&\qquad\qquad\left.\left.\text{$\gamma$ joining $b$ to
$c$}\right\}\right).
\end{split}
\end{equation}

\begin{theorem}[\textbf{Holonomic Parametrization of Self-financing Strategies}]
Let $M$ and $\mathcal{B}$ be connected. The Lie algebra
$\LAhol(\chi)$ parameterizes all homotopic self-financing strategies
in the market model in the following sense: it exists a group
isomorphism
\begin{equation}\phi:\LAhol(\chi)\rightarrow\Pi_1(X,D)\end{equation} mapping

\begin{itemize}
\item $0\in\LAhol(\chi)$ to the equivalence class of no arbitrage strategies which are null-contractible.
\item Non trivial elements of $\LAhol(\chi)$  to different equivalence classes of $\mathcal{D}$-admissible self-financing arbitrage strategies with the same start and end points.
\end{itemize}

\end{theorem}

\noindent This Theorem allows us to count the different
$\mathcal{D}$-admissible self-financing strategies which are
equivalent in terms of arbitrage. In fact the Lie algebra
$\LAhol(\chi)$ is mapped iso- and diffeomeorphically to the holonomy
Lie group $\Hol(\chi)$ by the exponential map. Therefore it follows

\begin{corollary}[\textbf{Number of Different Equivalent Self-Financing Strategies}]
The maps
\begin{equation}
\begin{CD}
\Hol(\chi)@>{\text{Exp}^{-1}}>>\LAhol(\chi)@>{\phi}>>\Pi_1(X,D)
\end{CD}
\end{equation}
are  group isomorphisms and manifold diffeomorphisms. In particular
\begin{equation}\left|\Pi_1(X,D)\right|=\left|\Hol(\chi)\right|.\end{equation}
\end{corollary}

\noindent This Corollary can be rephrased by saying that the
foliation by holonomy leaves of the market principal fiber bundle
$\mathcal{B}$ are in bijective correspondence with the equivalence
classes of $\mathcal{D}$-admissible self-financing arbitrage
strategies with the same start and end points. \par The results seen
so far in this section corroborate the belief that there is a deep
relationship between the market topology and the
no-free-lunch-with-vanishing-risk condition. As a matter of fact we
can complete Theorem \ref{holonomy} to

\begin{theorem}
The following assertions are equivalent:
\begin{itemize}
\item The market model satisfies the no-free-lunch-with-vanishing-risk condition.
\item The market homotopy group is trivial.
\item The market global holonomy group is trivial.
\item There exists a positive martingale and every point in the market has a neighborhood such that the arbitrage action vanishes for all closed strategies lying in that neighborhood.
\end{itemize}
\end{theorem}

The interpretation of this result is that for a market an arbitrage
possibility can only exist, if the market topology is not trivial.
That is true if and only if there are restrictions in the nominal
space acting as a topological obstruction, preventing \textit{every}
$\mathcal{D}$-admissible self-financing closed strategy from
being contractible.

\begin{example}[\textbf{Pension Funds' Market}]\text{}\\
Let us consider a market whose agents are all pension funds in the
world. The asset side of a pension fund is subject to several
regulatory constraints. Beside the usual no short sales constraints,
there are mixed linear constraints limiting the part of allocation
in  specific currencies, in regional markets, in the fixed-income or
equity market and so on. These restrictions translate into
hyperplanes cutting the nominal space $X$ and the market $M$. A
 situation as in Figure \ref{PF} can occur, where the pension fund
restrictions determine a ''hole'', that is, a non-trivial first
fundamental group: the strategy $\alpha$ is contractible but
strategy $\gamma$ not.
\end{example}
\begin{figure}[h!]
\centering
\includegraphics[width = 12.5cm, angle=-90]{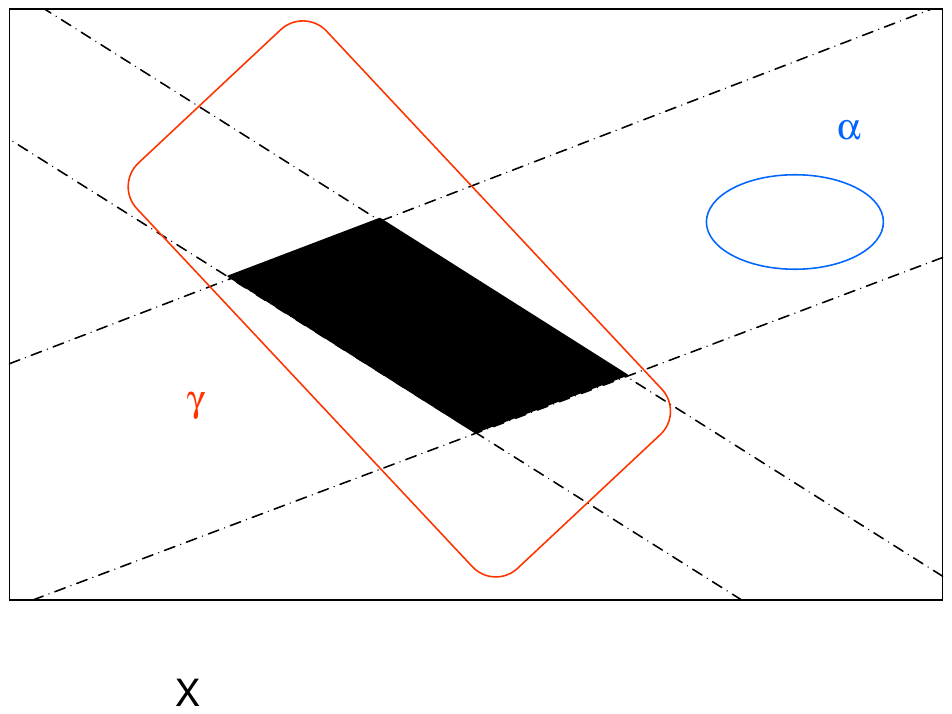}
\caption{Pension Funds' Market}\label{PF}
\end{figure}

\section{Lagrangian Theory of a Closed Market}\label{LagrangianTheory}

\subsection{Hamilton Principle and Lagrange Equations}\label{Hamilton}
We will now enforce the minimal arbitrage principle $(A4)$ for the
market portfolio strategy to derive the market dynamics for
deflators and short rates. There exists an interaction between
market portfolio and market dynamics: the market portfolio
allocation is determined by the choices of all market participants
and influence therefore the asset values. These -on their part- will
influence the choices of the market participants and therefore the
market portfolio. Everything happens simultaneously in time.
Summarizing:
\begin{equation}
\boxed{
\begin{array}{c}
  \text{Market dynamics for deflators and short rates } (D,r)=(D_t,r_t) \\
  \upharpoonleft\downharpoonright\\
  \text{Market Portfolio Strategy } x=x_t
\end{array}
}
\end{equation}

\noindent If the market is closed, that is, if there is no external
leverage, the market strategy must be self-financing. We denote by
$x_t$ the dynamics of the market portfolio and by $(D_t,r_t)$ that
of the deflators and short rates. By principle $(A4)$ these dynamics
are characterized by the fact that they must be a. s. minimizer of
the arbitrage action on the set of all self-financing strategies
which are candidates for the market portfolio. Remark that we do not
make any assumptions whether the market portfolio strategy allows
for arbitrage or not, we just assume that it is a minimizer.
The problem is intrinsically stochastic, but to tackle its solution by leveraging
on the techniques developed by Cresson and Darses (\cite{CrDa07}), we first analyze
the deterministic case in order to later construct a solution for the stochastic case.
To perform formal calculations we first introduce the Hamilton-Lagrange formalism of
classical mechanics and follow Chapter $3$ of the beautiful
\cite{Ar89}. We first treat the deterministic case and  embed the
 market portfolio strategy, deflators and short
rate dynamics into one parameter family strategies, deflators
and short rates. Then, we pass to stochastic processes following the results of \cite{CrDa07},
that we extend here to account for constraints to which the Lagrangian system is subject to.

\begin{definition}
Let $\gamma$ be the market $\mathcal{D}$-admissible strategy,
 and $\delta\gamma$, $\delta D$,
$\delta r$ be perturbations of the market strategy,
deflators' and short rates' dynamics. The \textbf{variation} of
$(\gamma,D,r)$ with respect to the given perturbations is the
following one parameter family:

\begin{equation}\epsilon\longmapsto(\gamma^\epsilon,D^\epsilon,r^\epsilon):=(\gamma,D,r)+ \epsilon(\delta \gamma,\delta D,\delta r).\end{equation}
Thereby, the parameter $\epsilon$ belongs to some open neighborhood
of $0\in\mathbf{R}$. The \textbf{arbitrage action} with respect to a positive martingale $\beta$ can be
consistently defined by
\begin{equation}
\begin{split}
A^{\beta}(\gamma;D,r)&:=\int_{\gamma}dt\BR{\mathcal{D}\rho^{\beta}+\dive_xJ}=\\
&=\int_{\gamma}dt\BR{\mathcal{D}\log(\beta_tD_t^{x_t})+r_t^x}=\\
&=\int_0^Tdt\frac{x_t\cdot \mathcal{D}D_t+x_t\cdot(r_tD_t)}{x_t\cdot D_t}+\log\frac{\beta_T}{\beta_0}.
\end{split}
\end{equation} and the first
variation of the arbitrage action as
\begin{equation}\delta A^{\beta}(\gamma;D,r):=\frac{d}{d\epsilon}A^{\beta}(\gamma^{\epsilon};D^{\epsilon},r^{\epsilon})\mid_{\epsilon:=0}.\end{equation}
\end{definition}
\noindent This leads to the following
\begin{definition}
Let us introduce the notation $q:=(x,D,r)$ and
$q^{\prime}:=(x^{\prime},D^{\prime},r^{\prime})$ for two vectors in
$\mathbf{R}^{3N}$. The \textbf{Lagrangian (or Lagrange function)} is
defined as
\begin{equation}L(q,q^{\prime}):=L(x,D,r,x^{\prime},D^{\prime},r^{\prime}):=\frac{x\cdot (D^{\prime}+ rD)}{x\cdot D}.\end{equation}

\end{definition}

\begin{lemma}
The arbitrage action for a self-financing
strategy $\gamma$ is the integral of the Lagrange function along the
$\mathcal{D}$-admissible strategy:
\begin{equation}
\begin{split}
A^{\beta}(\gamma;D,r)&=\int_{\gamma}dt\,L(q_t,q_t^{\prime})+\log\frac{\beta_1}{\beta_0}=\\
&=\int_{\gamma}dt\,L(x_t,D_t,r_t,x_t^{\prime},D_t^{\prime},r_t^{\prime})+\log\frac{\beta_1}{\beta_0}.
\end{split}
\end{equation}
\end{lemma}

\noindent A fundamental result of classical mechanics allows to
compute the extrema of the arbitrage action in the
\textit{deterministic} case as the solution of a system of ordinary
differential equations.

\begin{theorem}[\textbf{Hamilton Principle}]\label{ham}
Let us denote the derivative with respect to time as
$\frac{d}{dt}=:\prime$ and assume that all quantities observed are
deterministic. The local extrema of the arbitrage action satisfy the
Lagrange equations under the self-financing constraints
\begin{equation}
{\begin{split}
&{\delta A^{\beta}(\gamma;D,r)=0\text{ for all } (\delta\gamma,\delta D, \delta r)}\\
&\text{such that } {x_t^{\prime}}^{\epsilon}\cdot D^{\epsilon }_t=0 \text{ for all }\epsilon
\end{split}} \Leftrightarrow \left\{\begin{array}{ll}
                                      \frac{d}{dt}\frac{\partial L}{\partial x^{\prime}}-\frac{\partial L}{\partial x}=+2\lambda D^{\prime} \\
                                      \frac{d}{dt}\frac{\partial L}{\partial D^{\prime}}-\frac{\partial L}{\partial D}=-2\lambda x^{\prime} \\
                                      \frac{d}{dt}\frac{\partial L}{\partial r^{\prime}}-\frac{\partial L}{\partial r}=0 \\
                                      x^{\prime}\cdot D=0,
                                     \end{array}
                                    \right.\\
\end{equation}
\noindent where $\lambda\in\mathbf{R}$ denotes the constraint Lagrange multiplier.
\end{theorem}

\begin{remark}
By Corollary \ref{corHom}, if the variation does not include deflators and short rates, the arbitrage action is constant on every homotopic equivalence
class. Since we are looking for the optimal market strategy and asset market dynamics which minimizes the arbitrage, we have to vary over $\gamma$, $D$ and $r$ at the same time. This means that the integral of Lagrangian function takes values over a continuum and not over a discrete set as in the case of a fixed dynamics $D$ and $r$.
\end{remark}
\subsection{Stochastic Lagrangian Systems}\label{SLS} In this
subsection we briefly summarize and extend those contents of Cresson and Darses
(\cite{CrDa07}) needed in this paper. Cresson and Darses follow
previous works of Yasue (\cite{Ya81}) and Nelson (\cite{Ne01}).
\begin{defi}
Let $L=L(q,q^{\prime})$ be the Lagrange function of a deterministic
Lagrangian system with the non holonomic constraint $C(q,q^{\prime})=0$. Setting $L_{\lambda}:=L-\lambda C$ for the constraint Lagrange multipliers the dynamics is given by the extended Euler-Lagrange
equations
\begin{equation}
\boxed{\text{(EL)}\quad\left\{
         \begin{array}{l}
          \frac{d}{dt}\frac{\partial L_{\lambda}}{\partial q^\prime}(q,q^{\prime})-\frac{\partial L_{\lambda}}{\partial q}(q,q^{\prime})=0\\
          C(q,q^{\prime})=0
         \end{array}
       \right.
}
\end{equation}
meaning by this that the deterministic solution $q=q_t$ and $\lambda\in\mathbf{R}$ satisfy the constraint and
\begin{equation}\frac{d}{dt}\frac{\partial L_{\lambda}}{\partial q^\prime}\br{q_t, \frac{dq_t}{dt}}-\frac{\partial L_{\lambda}}{\partial
q}\br{q_t, \frac{dq_t}{dt}}=0.\end{equation} The formal
\textbf{stochastic embedding of the Euler-Lagrange equations} is
obtained by the formal substitution
\begin{equation}S:\frac{d}{dt}\longmapsto\mathcal{D},\end{equation}
and allowing the coordinates of the tangent bundle to be stochastic
\begin{equation}
\boxed{
\text{(SEL)}\quad\left\{
  \begin{array}{l}
    \mathcal{D}\frac{\partial L_{\lambda}}{\partial q^\prime}(q,q^{\prime})-\frac{\partial L_{\lambda}}{\partial q}(q,q^{\prime})=0\\
    C(q,q^{\prime})=0
  \end{array}
\right.
}
\end{equation}
meaning by this that the stochastic solution $Q=Q_t$ and the stochastic process $\lambda=\lambda_t$ satisfy the constraint and
\begin{equation}\mathcal{D}\frac{\partial L_{\lambda}}{\partial q^\prime}\br{Q_t, \mathcal{D}Q_t}-\frac{\partial L_{\lambda}}{\partial
q}\br{Q_t, \mathcal{D}Q_t}=0.\end{equation}
\end{defi}

\begin{defi}
Let $L=L(q,q^{\prime})$ be the Lagrange function of a deterministic
Lagrangian system on a time interval $I$ with constraint $C=0$. Set
\begin{equation}\Xi:=\BR{Q\in\mathcal{C}^1(I)\mid\mathbb{E}\Br{\int_I|L_{\lambda}(Q_t,\mathcal{D}Q_t)|dt}<+\infty}.\end{equation}
The action functional associated to $L_{\lambda}$ defined by
\begin{equation}
\boxed{
\begin{aligned}
&F: &\Xi\longrightarrow &\mathbf{R}\\
& &Q\longmapsto&\mathbb{E}\Br{\int_IL_{\lambda}(Q_t,\mathcal{D}Q_t)dt}
\end{aligned}
}
\end{equation}
 is called \textbf{stochastic analogue of the classic action} under the constraint $C=0$.
\end{defi}
For a sufficiently smooth extended Lagrangian $L_{\lambda}$ a necessary and sufficient
condition for a stochastic process to be a critical point of the
action functional $F$ is the fulfillment of the stochastic
Euler-Lagrange equations (SEL): see Theorem 7.1 page 54 in
\cite{CrDa07}. Moreover we have the following
\begin{lemma}[\textbf{Coherence}]
The following diagram commutes
\begin{equation}
\boxed{
    \xymatrix{
        L_{\lambda}(q_t,q_t^{\prime}) \ar[r]^{S} \ar[d]_{\text{Critical Action Principle}} & L_{\lambda}(Q_t,\mathcal{D}Q_t) \ar[d]^{\text{Stochastic Critical Action Principle}} \\
        (EL) \ar[r]_{S}       & (SEL) }
}
\end{equation}

\end{lemma}

\subsection{Arbitrage Dynamics For Deflators, Short Rates And Market
Portfolio}\label{arbEq}
 By means of the stochastization procedure
(see subsection \ref{SLS}) we can extend Theorem \ref{ham} to the
stochastic case:

\begin{theorem}[\textbf{Stochastic Hamilton Principle}]\label{SHP}
Let all quantities observed be stochastic and denote Nelson's
stochastic derivative with respect to time as $\mathcal{D}$. The
local extrema of the expected arbitrage action satisfy the Lagrange
equations under the self-financing constraints
\begin{equation}\label{systemSHE}
{\begin{split}
&\delta \mathbb{E}_0[A^{\beta}(\gamma,D,r)]=0\text{ for all } (\delta\gamma,\delta D, \delta r)\\
&\text{such that } \mathcal{D}x_t^{\epsilon}\cdot D^{\epsilon}_t= -\frac{1}{2}\mathfrak{D}_*\angles{x^{\epsilon},D^{\epsilon}}_t\text{ for all }\epsilon
\end{split}} \Leftrightarrow \left\{\begin{array}{ll}
                                      \mathcal{D}\frac{\partial L}{\partial x^{\prime}}-\frac{\partial L}{\partial x}=+\mathcal{D}(\lambda D) \\
                                      \mathcal{D}\frac{\partial L}{\partial D^{\prime}}-\frac{\partial L}{\partial D}=-\lambda \mathcal{D}x \\
                                      \mathcal{D}\frac{\partial L}{\partial r^{\prime}}-\frac{\partial L}{\partial r}=0 \\
                                      (\mathcal{D}x)\cdot D=-\frac{1}{2}\mathfrak{D}_*{\angles{x,D}},
                                     \end{array}
                                    \right.\\
\end{equation}
almost surely.
\end{theorem}

\noindent Before we tackle the problem of solving the stochastic
Euler-Lagrange equations,  we remark that they satisfy the gauge
invariance principle $(A3)$. As a matter of fact the Lagrange
function definition is invariant with respect to a coordinate change
in the tangent bundle  $T\Upsilon$, where $\Upsilon:=j^*(\mathcal{B})$ is the pullback of the market bundle with respect to the embedding
$j:X\rightarrow M,x\mapsto(x,0)$ (i.e. the market bundle without time component). In other words we can write
\begin{equation}L:T\Upsilon\rightarrow\mathbf{R}.\end{equation}
We will first solve the deterministic Euler-Lagrange equations and
then construct a stochastic solution by adding appropriate
perturbations with zero mean. More exactly, we write the stochastic
optimal solution as the sum of the deterministic one and a zero mean
perturbation $\delta x,\delta D, \delta r\in\mathcal{C}^1$ (see
Subsection \ref{SLS}) satisfying the conditions given by
(\ref{cond}).
\begin{equation}\label{def}
\begin{split}
x_t&=\mathbb{E}_0[x_t] + \delta x_t\\
D_t&=\mathbb{E}_0[D_t] + \delta D_t\\
r_t&=\mathbb{E}_0[r_t] + \delta r_t,
\end{split}
\end{equation}
whereas
\begin{equation}\label{cond}
\begin{matrix}
&\mathbb{E}_0[\delta x_t]=0, &\mathcal{D}\delta x_t=0, &\delta x_t\cdot \delta D_t=0,&\mathbb{E}_0[x_t]\cdot \delta r_t\delta D_t=0\\
&\mathbb{E}_0[\delta D_t]=0, &\mathcal{D}\delta D_t=0, &\mathbb{E}_0[x_t] \cdot \delta D_t=0,&(\mathbb{E}_0[r_t]\mathbb{E}_0[D_t])\cdot \delta x_t=0\\
&\mathbb{E}_0[\delta r_t]=0, &\mathcal{D}\delta r_t=0,
&\mathbb{E}_0[D_t]\cdot \delta x_t=0,&\delta x_t\cdot (\delta
r_t\delta D_t)=0\\
&\angles{\delta x_t,\delta D_t}=0. & & &
\end{matrix}
\end{equation}

\noindent We remark the Lagrange function satisfies
$L(q,q^\prime)=L(\mathbb{E}_0[q],\mathbb{E}_0[q^\prime])$ for any
$q=(x,D,r)$ satisfying conditions (\ref{def}) and(\ref{cond}).
Now we will compute for the arbitrage case the deterministic
solution of the Lagrange equations under the self-financing
constraints, which explicitly written out read
\begin{equation}\label{odesys}
\left\{
  \begin{split}
    &\left[\lambda_t(x_t\cdot D_t)^2+(x_t\cdot D_t)\right]D_t^\prime-\left[x_t\cdot D_t^{\prime}+x_t\cdot(r_tD_t)-\lambda_t^{\prime}(x_t\cdot D_t)^2\right]D_t+\\
    &\quad+(x_t\cdot D_t)(r_tD_t)=0\\
    &\left[\lambda_t(x_t\cdot D_t)^2-(x_t\cdot D_t)\right]x_t^{\prime}+\left[(x_t^\prime\cdot D_t)-x_t\cdot(r_tD_t)\right]x_t+\\
    &\quad+(x_t\cdot D_t)(x_tr_t)=0\\
    &x_tD_t=0\qquad\text{(this equation disappears if }r=0) \\
    &x_t^{\prime}\cdot D_t=0.
  \end{split}
\right.
\end{equation}
This system of ODE can be solved and we obtain the following result in the case of the classic model model without futures.
\begin{theorem}[\textbf{Arbitrage Market Dynamics}]\label{arbitrageDynamics}
Let $r\equiv0$. Then, the minimal arbitrage dynamics reads
\begin{equation}
\left\{
\begin{split}
x_t&=x_0+\delta x_t\\
D_t&=D_0+\delta D_t\\
\end{split}
\right.
\end{equation}
where $\delta x, \delta D$ and $\delta r$ are processes satisfying
condition (\ref{cond}). In particular, the expectationat time $0$ of the portfolio market nominals and the asset deflators are constant over time and equal to their initial values.
\end{theorem}
\begin{proof}
In the case $r\equiv0$ the system of ODEs (\ref{odesys}) becomes
\begin{equation}\label{odesys2}
\left\{
  \begin{split}
    &\left[\lambda_t(x_t\cdot D_t)^2+(x_t\cdot D_t)\right]D_t^\prime-\left[x_t\cdot D_t^{\prime}-\lambda_t^{\prime}(x_t\cdot D_t)^2\right]D_t=0\\
    &\left[\lambda_t(x_t\cdot D_t)^2-(x_t\cdot D_t)\right]x_t^{\prime}+\left[(x_t^\prime\cdot D_t)\right]x_t=0\\
    &x_t^{\prime}\cdot D_t=0.
  \end{split}
\right.
\end{equation}
It is a system of $2N+1$ first order ODEs in $2N+1$ unknown real values functions $(x,D,\lambda)$ of time.
We see that
\begin{equation}
\left\{
\begin{split}
x_t&\equiv x_0\\
D_t&\equiv D_0\\
\lambda_t&\equiv \lambda_0
\end{split}
\right.
\end{equation}
Since (\ref{odesys2}) is not a DAE system, by the Picard-Lindel\"of theorem we conclude that this solution is unique.
\end{proof}

\subsection{Symmetries, Conservation Laws and No-Arbitrage
Dynamics}\label{noarbitrage} In this subsection we continue the
study of the analogies between finance and classical mechanics with
the Hamilton-Lagrange formalism and introduce the concept of
symmetries of the asset model described by a Lagrangian on the
principle fibre bundle. By Chapter $4$ of \cite{Ar89} we will derive
from symmetries conservation laws which hold true in the general
case of an asset model allowing arbitrage or not.  The short rates
follow by the no-arbitrage principle by deriving the logarithm of
the solution for deflators.
\begin{defi}[\textbf{Market Symmetry}]
A bundle map $h:\Upsilon\rightarrow\Upsilon$ is called a
 symmetry of the market described by $(\Upsilon,L_\lambda)$ if and only if there exists a real $\tilde{\lambda}$ such that
\begin{equation}
L_{\lambda}(T_bh.B)=L_{\tilde{\lambda}}(B)\quad\text{ for all }B\in T_b\Upsilon\,,\text{ for all } b\in\Upsilon.\end{equation}
\end{defi}
\begin{example}[\textbf{Market Symmetries}]\text{}\\
We represent term structures by their short rate, so that
$b=(x,D,r)$ and $h=h(x,D,r)$.
\begin{itemize}
\item \textbf{Rotation:} $h(x,D,r):=[Sx;SD;S(rD)/(rD)]$, where $S$ is a
 orthogonal matrix, i.e. $S\in O(N)$. The division of two vectors is meant componentwise.
\item \textbf{Nominal Dilation:} $h(x,D,r):=[sx;D;r]$, where
$s$ is a non vanishing real number, i.e. $s\in \mathbf{R}^*$.
\item \textbf{Deflator Dilation:} $h(x,D,r):=[x; sD;r]$, where
$s$ is a non vanishing real number, i.e. $s\in \mathbf{R}^*$.
\end{itemize}
These examples all fulfill the definition of market symmetry, as one
can see from the Lagrange function
\begin{equation}L_{\lambda}(\underbrace{x,D,r}_{=q},\underbrace{x^{\prime},D^{\prime},r^{\prime}}_{=q^{\prime}})=\frac{x\cdot (D^{\prime}+ rD)}{x\cdot D}-\lambda(x^{\prime}\cdot D).\end{equation}
\end{example}

\noindent The connection between symmetries and conservation laws in
classical mechanics can be restated for our market model.
\begin{theorem}[\textbf{Noether}]\label{DetNoet}
Assume that all quantities observed are deterministic. Let
$\{h_\epsilon\}_{\epsilon}$ be a one-parameter group of market
symmetries $h_\epsilon:\Upsilon\rightarrow\Upsilon$ for the
market described by $(\Upsilon,L_{\lambda})$. Then, the dynamics of market
the portfolio, deflators and short rates have a first integral
$I:T\Upsilon\rightarrow\mathbf{R}$. This means that there is a
function, which writes
\begin{equation}I(q,q^{\prime})=\frac{\partial L_{\lambda}}{\partial q^ {\prime}}\frac{d h_\epsilon(q)}{d \epsilon}\mid_{\epsilon=0},\end{equation}
such that
\begin{equation}\frac{d}{dt}I(q_t,q^{\prime}_t)=0\end{equation}
where $q=q_t$ is the solution of the \textit{deterministic}
Euler-Lagrange equations.
\end{theorem}

By means of the stochastization procedure as explained in
Subsection \ref{SLS}, we can extend the preceding Theorem to the
stochastic case (see \cite{CrDa07} page 60):

\begin{theorem}[\textbf{Stochastic Version of Noether's Result}]\label{StochNoet}
Let all quantities observed be stochastic and
$\{h_\epsilon\}_{\epsilon}$ be a one-parameter group of market
symmetries $h_\epsilon:\Upsilon\rightarrow\Upsilon$ for the
market described by $(\Upsilon,L_{\lambda})$. Then, the dynamics of market
the portfolio, deflators and short rates have a first integral
$I:T\Upsilon\rightarrow\mathbf{R}$. This means that there is a
function, which writes
\begin{equation}I(q,q^{\prime})=\mathbb{E}_0\Br{\frac{\partial L_{\lambda}}{\partial q^ {\prime}}\frac{d h_\epsilon(q)}{d \epsilon}\mid_{\epsilon=0}},\end{equation}
such that
\begin{equation}\frac{d}{dt}I(Q_t,\mathcal{D}Q_t)=0\end{equation}
where $Q=Q_t$ is the solution of the \textit{stochastic}
Euler-Lagrange equations.
\end{theorem}

\noindent In classical mechanics Noether's Theorem is applied to a
$N$ point particle system to derive the so called $10$ conservation
laws for energy, momentum, angular momentum, and center of mass of
an isolated system, i.e. with no external forces acting on its
particles. Now we will derive conservation laws from the
symmetries of a market with $N$ assets and no external leverage.
That for, we have to extend the symmetries of our example to one
parameter groups.

\begin{example}[\textbf{One Parameter Group of Market Symmetries}]\text{}\label{ExNoet}

\begin{itemize}
\item \textbf{Rotations' Group:} $h_\epsilon(x,D,r):=[S_\epsilon x;S_\epsilon D;S_{\epsilon}(rD)/(rD))]$, where $S_\epsilon$ is a
 one parameter group of orthogonal matrices, i.e. $S_\epsilon\in O(N)$, and $S_0=I_N$, the identity matrix in $N$ dimensions. Therefore,
 $h_0(x,D,r)=[x;D;r]$. The first integral is
 \begin{equation}I=\mathbb{E}_0\Br{\frac{\partial L_{\lambda}}{\partial q^ {\prime}}\Xi[x;D;e]},\end{equation}
 where $\Xi:=\diag(\xi,\xi,\xi)$ is a $3N \times 3N$ matrix,  $\xi$
 is a $N \times N$ antisymmetric matrix and $e=[1,\dots,1]^{\dagger}$. By Theorem (\ref{StochNoet}) we
 obtain $(N-1)N/2$ non-trivial first integrals:
\begin{equation}
  \mathbb{E}_0\left[-\lambda_tD_t\cdot\xi x_t+\frac{x_t}{x_t\cdot D_t}\cdot \xi x_t\right]\equiv\text{ Time Constant.}
\end{equation}

\item \textbf{Nominal Dilations' Group:} $h_\epsilon(x,D,r):=[(1+ \epsilon)x;D;r]$, where
$\epsilon$ is real number, so that $h_0(x,D,r)=[x;D;r]$.
The first integral is
 \begin{equation}I=\mathbb{E}_0\Br{\frac{\partial L_{\lambda}}{\partial q^ {\prime}}[x;0;0]},\end{equation}
By Theorem (\ref{StochNoet}) we obtain one non-trivial first
integral:
\begin{equation}\mathbb{E}_0\left[-\lambda_t(x_t\cdot D_t)\right]\equiv\text{ Time Constant.}\end{equation}

\item \textbf{Deflator Dilations' Group:} $h_\epsilon(x,D,r):=[x;(1+\epsilon)D;r]$, where
$\epsilon$ is real number, so that
$h_0(x,D,r)=[x;D;r]$. The first integral is
 \begin{equation}I=\mathbb{E}_0\Br{\frac{\partial L_{\lambda}}{\partial q^ {\prime}}[0;D;0]}.\end{equation}By Theorem (\ref{StochNoet}) we obtain
 one trivial first integral, the $1$-constant function:
\begin{equation}\mathbb{E}_0\Br{\frac{x_t\cdot D_t}{x_t\cdot D_t}}\equiv 1.\end{equation}

\item \textbf{Time Translations' Group:} since the time has been excluded by construction from the bundle $\Upsilon$, we cannot utilize N\"other's result as depicted in Theorems \ref{DetNoet} and \ref{StochNoet}. Therefore, we compute, for the deterministic case
    \begin{equation}
     \begin{split}
      \frac{dL_{\lambda}}{dt}&=\frac{\partial L_{\lambda}}{\partial q}\cdot q^{\prime}+\frac{\partial L_{\lambda}}{\partial q^{\prime}}\cdot q^{\prime\prime}+\frac{\partial L_{\lambda}}{\partial t}=\\
      &=\frac{\partial L_{\lambda}}{\partial q}\cdot q^{\prime}+\frac{d}{dt}\left[\frac{\partial L_{\lambda}}{\partial q^{\prime}}\cdot q^{\prime}\right]-\frac{d}{dt}\left[\frac{\partial L_{\lambda}}{\partial q^{\prime}}\right]\cdot q^{\prime}+\frac{\partial L_{\lambda}}{\partial t}.
     \end{split}
    \end{equation}
    Since the Euler-Lagrange equations are fulfilled, we obtain
    \begin{equation}
    \frac{d}{dt}\left[\frac{\partial L_{\lambda}}{\partial q^{\prime}}\cdot q^{\prime}-L_{\lambda}\right]=-\frac{\partial L_{\lambda}}{\partial t}.
    \end{equation}
    Time translation invariance of $L_{\lambda}$ means $\frac{\partial L_{\lambda}}{\partial t}=0$, and, therefore
    \begin{equation}
    \frac{\partial L_{\lambda}}{\partial q^{\prime}}\cdot q^{\prime}-L_{\lambda}\equiv\text{ Time Constant,}
    \end{equation}
    and in the stochastic case
    \begin{equation}
    \mathbb{E}_0\left[\frac{\partial L_{\lambda}}{\partial q^{\prime}}\cdot q^{\prime}-L_{\lambda}\right]\equiv\text{ Time Constant,}
    \end{equation}
    which in our case reads
    \begin{equation}
    \mathbb{E}_0\left[-\frac{x_t\cdot (r_tD_t)}{(x_t\cdot D_t)}\right]\equiv\text{ Time Constant.}
    \end{equation}
\end{itemize}

\end{example}

We can now utilize the results from the preceding example to complete Theorem \ref{arbitrageDynamics} for the general case, where forwards are allowed as assets.

\begin{theorem}[\textbf{No Arbitrage Market
Dynamics}]\label{noarbitrageDynamics2} In a closed market satisfying
the  no-free-lunch-with-vanishing-risk condition, the dynamics for
market portfolio strategy, deflators and term structures have constant expectations
over time. More exactly the following identity holds a.s.
\begin{equation}
\begin{split}
x_t&= x_0 + \delta x_t\\
D_t&= D_0 + \delta D_t\\
r_t&= r_0 + \delta r_t
\end{split}
\end{equation}
where $\delta x, \delta D$ are processes satisfying condition
(\ref{cond}).
\end{theorem}
\begin{proof}
Let us consider the deterministic case and enrich the system of ODEs \ref{odesys} with the equations obtained by Example \ref{ExNoet}.
After some computations we obtain
\begin{equation}
\left\{
\begin{split}
&x_t\cdot(D_t^{\prime}+r_tD_t)=0\\
&x_t\cdot(\lambda_tr_tD_t-\lambda_t^{\prime}D)=0\\
&x_t^{\prime}\cdot D_t=0\\
&\lambda_t (x_t\cdot D_t)\equiv\text{const}\\
&\frac{x_t\cdot (r_tD_t)}{(x_t\cdot D_t)}\equiv\text{const}\\
&-\lambda_tD_t\xi x_t+\frac{x_t\cdot \xi x_t}{(x_t\cdot D_t)}\equiv\text{const},
\end{split}
\right.
\end{equation}
\noindent where $\xi$ is an arbitrary antisymmetric $N\times N$ matrix. By differentiating the equations where on the r.h.s. there is an unknown constant we get
\begin{equation}\label{eqxi}
\left\{
\begin{split}
&x_t\cdot(D_t^{\prime}+r_tD_t)=0\\
&x_t\cdot(\lambda_tr_tD_t-\lambda_t^{\prime}D)=0\\
&x^{\prime}\cdot D_t=0\\
&\lambda_t^{\prime} (x_t\cdot D_t)+\lambda_t (x_t^{\prime}\cdot D_t)+\lambda_t (x_t\cdot D_t^{\prime})=0\\
&\left[x_t^{\prime}\cdot (r_tD_t)+x_t\cdot (r_t^{\prime}D_t)+x_t\cdot (r_tD_t^{\prime})\right](x_t\cdot D_t)+\\
&\qquad-x_t\cdot (r_tD_t)\left[(x_t^{\prime}\cdot D_t)+(x_t\cdot D_t^{\prime})\right]=0\\
&(\lambda_t^{\prime}D_t\cdot \xi x_t-\lambda_tD_t^{\prime}\xi x_t)(x_t\cdot D_t)^2+(x_t^{\prime}\cdot \xi x_t+x_t\cdot \xi x^{\prime})(x_t\cdot D_t)+\\
&\qquad-(x_t\cdot\xi x_t)(x_t\cdot D_t^{\prime})=0,
\end{split}
\right.
\end{equation}
By an appropriate choice of $\xi$ the system (\ref{eqxi}) becomes a system  of $3N+1$ first order ODEs in $3N+1$ unknown real values functions $(x,D,r, \lambda)$ of time.
We see that
\begin{equation}
\left\{
\begin{split}
x_t&\equiv x_0\\
D_t&\equiv D_0\\
r_t&\equiv r_0\\
\lambda_t&\equiv \lambda_0
\end{split}
\right.
\end{equation}
Since (\ref{eqxi}) is not a DAE system, by the Picard-Lindel\"of theorem we conclude that this solution is unique. Moreover, it fullfills the last equation of (\ref{eqxi}) for any $\xi$. The proof is completed.\\
\end{proof}

\section{Conclusion}
By introducing an appropriate stochastic differential geometric
formalism the classical theory of stochastic finance can be
embedded into a conceptual framework called \textit{Geometric
Arbitrage Theory}, where the market is modelled with a principal
fibre bundle, arbitrage corresponds to its curvature and arbitrage
strategies to its holonomy. The Fundamental Theorem of Asset
Pricing is given a differential homotopic characterization. The
market dynamics is seen to be the solution of stochastic
Euler-Lagrange equations for a choice of the Lagrangian allowing
to express Hamilton's principle of minimal action as the minimal
expected arbitrage principle, an extension of the no-arbitrage
principle. Explicit are
provided for a closed market.

\section*{Acknowledgements}
We would like to extend our gratitude to Hideyuki Takada, Gianni Arioli, Giovanni Paolinelli, Yuri Gliklikh, Josef Teichmann, Juan Pablo Ortega, Freddy
Delbaen,  Ren\'e Carmona, Lee Smolin, Robert Sch\"{o}ftner,
Samuel Vazquez, Mario Clerici, and Simone Severini  for many
discussions and ideas which influenced the results of this paper.


\begin{thebibliography}{99}
\bibitem[Ar89]{Ar89}
  \newblock V. I. Arnold,
  \newblock  \emph{Mathematical Methods of Classical Mechanics},
  \newblock Graduate Texts in Mathematics, Second Edition, Springer 1989.

\bibitem[BeFr02]{BeFr02}
  \newblock F. Bellini and M. Frittelli,
  \newblock  \emph{On the Existence of Minimax Martingale Measures},
  \newblock Mathematical Finance, \textbf{12/1}, (1-21), 2002.

\bibitem[Bj04]{Bj04}
  \newblock T. Bj\"ork,
  \newblock  \emph{Arbitrage Theory in Continuous Time},
  \newblock Oxford Finance, Second Edition, 2004.

\bibitem[BjHu05]{BjHu05}
  \newblock T. Bj\"ork and H. Hult,
  \newblock  \emph{A Note on Wick Products and the Fractional Black-Scholes Model},
  \newblock Finance \& Stochastics, \textbf{9}, No.2, (197-209), 2005.

\bibitem[Bl81]{Bl81}
  \newblock D. Bleecker,
  \newblock  \emph{Gauge Theory and Variational Principles},
  \newblock Addison-Wesley Publishing, 1981, (republished by Dover 2005).

\bibitem[CrDa07]{CrDa07}
  \newblock J. Cresson and S. Darses,
  \newblock  \emph{Stochastic Embedding of Dynamical Systems},
  \newblock J. Math. Phys. \textbf{48}, 2007.

\bibitem[DeSc08]{DeSc08}
  \newblock F. Delbaen and W. Schachermayer,
  \newblock  \emph{The Mathematics of Arbitrage},
  \newblock Springer 2008.

\bibitem[DuFoNo84]{DuFoNo84}
  \newblock B. A. Dubrovin, A. T. Fomenko and S. P. Novikov,
  \newblock  \emph{Modern Geometry-Methods and Applications: Part II. The Geometry and Topology of Manifolds},
  \newblock Springer GTM, 1984.

\bibitem[DuFiMu00]{DuFiMu00}
  \newblock B. Dupoyet, H. R. Fiebig and D. P. Musgrov,
  \newblock  \emph{Gauge Invariant Lattice Quantum Field Theory: Implications for Statistical Properties of High Frequency Financial Markets},
  \newblock Physica A \textbf{389} (107-116), 2010.

\bibitem[DeMe80]{DeMe80}
  \newblock C. Dellach\'{e}rie and P. A.  Meyer,
  \newblock  \emph{Probabilit\'{e} et potentiel II - Th\'{e}orie des martingales - Chapitres 5 \`{a} 8},
  \newblock Hermann, 1980.

\bibitem[El82]{El82}
  \newblock K. D. Elworthy,
  \newblock  \emph{Stochastic Differential Equations on Manifolds},
  \newblock London Mathematical Society Lecture Notes Series, 1982.

\bibitem[Em89]{Em89}
  \newblock M. Em\'{e}ry,
  \newblock  \emph{Stochastic Calculus on Manifolds-With an Appendix by P. A. Meyer},
  \newblock Springer, 1989.

\bibitem[Fa15]{Fa15}
  \newblock S. Farinelli,
  \newblock  \emph{Geometric Arbitrage Theory and Market Dynamics},
  \newblock Journal of Geometric Mechanics, \textbf{7 (4)}, (431-471), 2015.

\bibitem[FaVa12]{FaVa12}
  \newblock S. Farinelli and S. Vazquez,
  \newblock  \emph{Gauge Invariance, Geometry and Arbitrage},
  \newblock The Journal of Investment Strategies, Volume 1/Number 2, Wiley, (23-66), Spring 2012.

\bibitem[FeJi07]{FeJi07}
  \newblock M. Fei-Te and M. Jin-Long,
  \newblock  \emph{Solitary Wave Solutions of Nonlinear Financial Markets: Data-Modeling-Concept-Practicing},
  \newblock Front. Phys. China, \textbf{2(3)},(368-374), 2007.

\bibitem[FlHu96]{FlHu96}
  \newblock B. Flesaker and L. Hughston,
  \newblock  \emph{Positive Interest},
  \newblock Risk \textbf{9 (1)}, (115-124), 1996.

\bibitem[F\"oSc04]{FoeSc04}
  \newblock H. F\"ollmer and A. Schied,
  \newblock  \emph{Stochastic Finance: An Introduction In Discrete Time},
  \newblock Second Edition, De Gruyter Studies in Mathematics, 2004.

\bibitem[Gl11]{Gl11}
  \newblock Y. E. Gliklikh,
  \newblock  \emph{Global and Stochastic Analysis with Applications to Mathematical Physics},
  \newblock Theoretical and Mathemtical Physics, Springer, 2011.


\bibitem[HaTh94]{HaTh94}
 \newblock W. Hackenbroch and A. Thalmaier,
  \newblock  \emph{Stochastische Analysis. Eine Einf\"{u}hrung in die Theorie der stetigen Semimartingale},
\newblock Teubner Verlag, 1994.

\bibitem[H\"o03]{Ho03}
  \newblock  L. H\"ormander,
  \newblock  \emph{The Analysis of Linear Partial Differential Operators I: Distribution Theory and Fourier Analysis},
  \newblock Springer, 2003.

\bibitem[Hs02]{Hs02}
  \newblock  E. P. Hsu,
  \newblock  \emph{Stochastic Analysis on Manifolds},
  \newblock Graduate Studies in Mathematics, \textbf{38}, AMS, 2002.

\bibitem[HuKe04]{HuKe04}
  \newblock P. J. Hunt and J. E. Kennedy,
  \newblock  \emph{Financial Derivatives in Theory and Practice},
  \newblock Wiley Series in Probability and Statistics, 2004.


\bibitem[Il00]{Il00}
  \newblock  K. Ilinski,
  \newblock  \emph{Gauge Geometry of Financial Markets},
  \newblock J. Phys. A: Math. Gen. \textbf{33}, (L5-L14), 2000.

\bibitem[Il01]{Il01}
  \newblock  K. Ilinski,
  \newblock  \emph{Physics of Finance: Gauge Modelling in Non-Equilibrium Pricing},
  \newblock Wiley, 2001.

\bibitem[Ja98]{Ja98}
  \newblock  J. D. Jackson,
  \newblock  \emph{Classical Electrodynamics},
  \newblock Third Edition, Wiley, 1998.

\bibitem[KoNo96]{KoNo96}
  \newblock S. Kobayashi and K. Nomizu,
  \newblock  \emph{Foundations of Differential Geometry, Volume I},
  \newblock Wiley, 1996.

\bibitem[Ma96]{Ma96}
  \newblock P. N. Malaney,
  \newblock  \emph{The Index Number Problem: A Differential Geometric Approach},
  \newblock PhD Thesis, Harvard University Economics Department, 1996.

\bibitem[Mo09]{Mo09}
  \newblock Y. Morisawa,
  \newblock  \emph{Toward a Geometric Formulation of Triangular Arbitrage: An Introduction to Gauge Theory of Arbitrage},
  \newblock Progress of Theoretical Physics Supplement \textbf{179},  (209-214), 2009.

\bibitem[Ne01]{Ne01}
  \newblock E. Nelson,
  \newblock  \emph{Dynamical Theories of Brownian Motion},
  \newblock Second Edition, Princeton, 2001.

\bibitem[Pr10]{Pr10}
  \newblock Ph. E. Protter,
  \newblock  \emph{Stochastic Integration and Differential Equations: Version 2.1},
  \newblock Stochastic Modelling and Applied Probability, Springer, 2010.

\bibitem[Ro94]{Ro94}
  \newblock L. C. G. Rogers,
  \newblock  \emph{Equivalent Martingale Measures and No-Arbitrage},
  \newblock  Stochastics, Stochastics Rep. \textbf{51}, (41-49), 1994.

\bibitem[Scha01]{Scha01}
  \newblock W. Schachermayer,
  \newblock  \emph{Optimal Investment in Incomplete Markets When Wealth May Become Negative},
  \newblock  Annals of Applied Probability, \textbf{11}, No. 3, (694-734), 2001.

\bibitem[Schw80]{Schw80}
  \newblock  L. Schwartz,
  \newblock  \emph{Semi-martingales sur des vari\'{e}t\'{e}s et martingales conformes sur des vari\'{e}t\'{e}s analytiques complexes},
  \newblock  Springer Lecture Notes in Mathematics, 1980.

\bibitem[Sh00]{Sh00}
  \newblock  S. E. Shreve,
  \newblock  \emph{Stochastic Calculus for Finance},
  \newblock  Vol. II, Springer, 2000.

\bibitem[\v{S}i02]{Si02}
  \newblock  M. \v{S}ilhav\'{y},
  \newblock  \emph{The Mechanics and Thermodynamics of Continuous Media},
  \newblock  Theoretical and Mathematical Physics, Springer, 2002.

\bibitem[SmSp98]{SmSp98}
  \newblock  A. Smith and C. Speed,
  \newblock  \emph{Gauge Transforms in Stochastic Investment},
  \newblock  Proceedings of the 1998 AFIR Colloquim, Cambridge, England, 1998.

\bibitem[St82]{St82}
  \newblock  S. Sternberg,
  \newblock  \emph{Lectures On Differential Geometry},
  \newblock  Second Edition,  Chelsea Pub. Co., 1982.

\bibitem[St00]{St00}
  \newblock   D. W. Stroock,
  \newblock  \emph{An Introduction to the Analysis of Paths on a Riemannian Manifold},
  \newblock  Mathematical Surveys and Monographs, \textbf{74}, AMS, 2000.

\bibitem[We06]{We06}
  \newblock   E. Weinstein,
  \newblock  \emph{Gauge Theory and Inflation: Enlarging the Wu-Yang Dictionary to a unifying Rosetta Stone for Geometry in Application},
  \newblock  Talk given at Perimeter Institute, 2006.

\bibitem[Ya81]{Ya81}
  \newblock   K. Yasue,
  \newblock  \emph{Stochastic Calculus of Variations},
  \newblock Journal of Functional Analysis \textbf{41}, (327-340), 1981.

\bibitem[Yo99]{Yo99}
  \newblock  K. Young,
  \newblock  \emph{Foreign Exchange Market as a Lattice Gauge Theory},
  \newblock  Am. J. Phys. \textbf{67}, (862-868), 1999.

\end{thebibliography}
\end{document}